\documentclass[11pt]{article}
\usepackage{amsmath,amssymb,amsthm}
\usepackage{enumerate}
\usepackage{hyperref}
\usepackage{mathabx}
\usepackage{blkarray}
\usepackage{tikz}
\usepackage{authblk}
\usepackage{graphicx}
\usepackage{xspace}
\usepackage[margin=1.2in,letterpaper]{geometry}

\usetikzlibrary{backgrounds}
\usetikzlibrary{shapes}
\usetikzlibrary{positioning}
\usepackage{todonotes}
\hypersetup{
	pdftitle = {2019ptolemaic_deletion},
	pdfauthor = {Jungho Ahn, Eun Jung Kim, and Euiwoong Lee}
}
\newcommand\abs[1]{\lvert #1\rvert}

\usepackage{amsmath}
\newtheorem{THM}{Theorem}[section]
\newtheorem{LEM}[THM]{Lemma}

\newtheorem{PROP}[THM]{Proposition}

\newtheorem{OBS}[THM]{Observation}

\newtheorem{CLAIM}{Claim}

\theoremstyle{remark}

\newenvironment{proofofclaim}{\begin{list}{}{
              \setlength{\leftmargin}{0mm}
              } \item {\it Proof of Claim.}}{\hfill$\blacklozenge$\end{list}\medskip}

\theoremstyle{definition}

\newcommand{\ora}{\overrightarrow}
\newcommand{\anc}{\textit{anc}}
\newcommand{\des}{\textit{des}}
\newcommand{\rt}{\textit{src}}
\newcommand{\cev}[1]{\reflectbox{\ensuremath{\vec{\reflectbox{\ensuremath{#1}}}}}}

\newcommand{\polylog}{\mathrm{polylog}}
\newcommand{\R}{\mathbb{R}}
\newcommand{\C}{\mathcal{C}}
\newcommand{\M}{\mathcal{M}}
\newcommand{\Z}{\mathcal{Z}}

\newcommand{\F}{\mathcal{F}}
\newcommand{\I}{\mathcal{I}}
\newcommand{\iphi}{\phi^{-1}}
\newcommand{\RR}{\mathcal{R}}

\newcommand{\bx}{\bar{x}}
\newcommand{\eps}{\varepsilon}

\newcommand{\opt}{\mathsf{OPT}}

\newcommand{\lp}{\mathsf{LP}}
\newcommand{\und}{\mathsf{Und}}
\newcommand{\gem}{{\sf gem}\xspace}
\newcommand{\bull}{{\sf bull}\xspace}
\newcommand{\dart}{{\sf dart}\xspace}
\newcommand{\dia}{{\sf diamond}\xspace}

\newcommand{\closed}{downward-closed}

\newcommand{\dG}{\ora{G}}
\newcommand{\dH}{\ora{H}}

\newcommand{\dT}{\ora{T}}

\newcommand{\PD}{\textsc{Ptolemaic Deletion}\xspace}
\newcommand{\WPD}{\textsc{Weighted Ptolemaic Deletion}\xspace}
\newcommand{\FD}{$\mathcal{F}$-\textsc{Deletion}\xspace}
\newcommand{\WFD}{\textsc{Weighted} $\mathcal{F}$-\textsc{Deletion}\xspace}
\newcommand{\VC}{\textsc{Vertex Cover}\xspace}

\newcommand{\OCT}{\textsc{Odd Cycle Transversal}\xspace}

\newcommand{\FVSs}{\textsc{FVS}\xspace}
\newcommand{\FVS}{\textsc{Feedback Vertex Set}\xspace}
\newcommand{\WFVS}{\textsc{Weighted Feedback Vertex Set}\xspace}
\newcommand{\SFVS}{\textsc{Subset Feedback Vertex Set}\xspace}

\newcommand{\FVSPs}{\textsc{FVSP}\xspace}

\newcommand{\FVSP}{\textsc{Feedback Vertex Set with Precedence Constraints}\xspace}

\newcommand{\HVC}{\textsc{Hypergraph Vertex Cover}\xspace}

\newcommand{\pto}{ptolemaic\xspace}

\usepackage{boxedminipage}

\begin{document}

\title{Towards constant-factor approximation for chordal / distance-hereditary vertex deletion}
\author[1]{Jungho Ahn}
\author[2]{Eun Jung Kim}
\author[3]{Euiwoong Lee}
\affil[1]{Department of Mathematical Sciences, KAIST, Daejeon, South Korea}
\affil[1]{Discrete Mathematics Group, Institute for Basic Science (IBS), Daejeon, South Korea}
\affil[2]{Université publique, CNRS, Paris, France}
\affil[3]{Department of Computer Science, New York University, New York, USA}
\affil[ ]{\small \textit{Email addresses:} \texttt{junghoahn@kaist.ac.kr}, \texttt{eun-jung.kim@dauphine.fr}, \texttt{euiwoong@cs.nyu.edu}}
\renewcommand{\thefootnote}{\roman{footnote}}
\footnotetext[1]{Jungho Ahn is supported by the Institute for Basic Science (IBS-R029-C1). Eun Jung Kim is supported by ANR JCJC project ``ASSK'' (ANR-18-CE40-0025-01). Euiwoong Lee is supported by Simons Collaboration on Algorithms and Geometry.}
\footnotetext[2]{A part of this research was done during the ``2019 IBS Summer Research Program on Algorithms and Complexity in Discrete Structures'', hosted by the IBS Discrete Mathematics Group.}

\maketitle

\renewcommand{\thefootnote}{\arabic{footnote}}

\begin{abstract}\label{Abstract}
For a family of graphs $\F$, \WFD is the problem for which the input is a vertex weighted graph $G = (V, E)$ and the goal is to delete $S \subseteq V$ with minimum weight such that $G \setminus S \in \F$.
Designing a constant-factor approximation algorithm for large subclasses of perfect graphs has been an interesting research direction.
Block graphs, 3-leaf power graphs, and interval graphs are known to admit constant-factor approximation algorithms, but the question is open for chordal graphs and distance-hereditary graphs.

In this paper, we add one more class to this list by presenting a constant-factor approximation algorithm when $\F$ is the intersection of chordal graphs and distance-hereditary graphs.
They are known as {\em \pto graphs} and form a superset of both block graphs and 3-leaf power graphs above.
Our proof presents new properties and algorithmic results on {\em inter-clique} digraphs as well as an approximation algorithm for a variant of \FVS that exploits this relationship (named \FVSP), each of which may be of independent interest.
\end{abstract}

\section{Introduction}\label{sec: introduction}

Given a family of graphs $\F$, we consider the following problem.

\vskip 0.2cm
\noindent
\fbox{\parbox{0.97\textwidth}{
\textsc{\WFD}\\
\textbf{Input :} A graph $G = (V, E)$ with vertex weights $w : V \to \R^+ \cup \{ 0 \}$. \\
\textbf{Question :} Find a set $S \subseteq V$ of minimum weight such that $G \setminus S \in \F$. 
}}
\vskip 0.2cm

This problem captures many classical combinatorial optimization problems including \VC, \FVS, \OCT, and the problems corresponding to natural graph classes (e.g., planar graphs, chordal graphs, or graphs of bounded treewidth) also have been actively studied.
Most of these problems, including the simplest \VC, are NP-hard, so polynomial-time exact algorithms are unlikely to exist for them.

Parameterized algorithms and approximation algorithms have been two of the most popular kinds of algorithms for NP-hard optimization problems, and \FD has been actively studied from both viewpoints.
There is a large body of work in the theory of parameterized complexity, where \FD for many $\F$'s is shown to be in FPT or even admits a polynomial kernel.
The list of such $\F$'s includes chordal graphs~\cite{marx2010chordal,jansen2017approximation,agrawal2018feedback}, interval graphs~\cite{cao2015interval,cao2016linear,agrawal2019interval}, distance-hereditary graphs~\cite{eiben2018single, kim2017polynomial}, bipartite graphs~\cite{reed2004finding, kratsch2014compression}, and graphs with bounded treewidth~\cite{fomin2012planar,kim2015linear}.

On the other hand, despite large interest, approximability for \FD is not as well as understood as parameterized complexity.
To the best of our knowledge, for all $\F$'s admitting parameterized algorithms in the above paragraph except \OCT, the existence of a constant-factor approximation algorithm is not ruled out under any complexity hypothesis.
When $\F$ can be characterized by a finite list of forbidden subgraphs or induced subgraphs (not minors), the problem becomes a special case of \HVC with bounded hyperedge size, which admits a constant-factor approximation algorithm.
Besides them, the only classes of graphs that currently admit constant-factor approximation algorithms are block graphs~\cite{agrawal2016faster}, 3-leaf power graphs~\cite{AEKO2019}, interval graphs~\cite{cao2016linear}, and graphs of bounded treewidth~\cite{fomin2012planar, gupta2019losing}.
Weighted versions are sometimes harder than their unweighted counterparts, and within graphs of bounded treewidth, the only two nontrivial classes whose weighted version admits a constant-factor approximation algorithm are the set of forests (\WFVS) and the set of graphs excluding a \dia as a minor~\cite{fiorini2010hitting}.
See Figure~\ref{Figure: A bull, a dart, a diamond, and a gem}.

When $\F$ is the set of perfect or weakly chordal graphs, it is known that a constant-factor approximation algorithm is unlikely to exist~\cite{heggernes2011parameterized}.
Therefore, there has been recent interest on identifying large subclasses of perfect graphs that admit constant-factor approximation algorithms.
Among the subclasses of perfect graphs, chordal graphs and distance-hereditary graphs have drawn particular interest.
Recall that chordal graphs are the graphs without any induced $C_{\geq 4}$\footnote{Let $C_{\geq k}$ be the set of cycles of length at least $k$.}, and distance-hereditary graphs are the graphs without any induced $C_{\geq 5}$, a \gem, a {\sf house}, or a {\sf domino}.
See Figure~\ref{Figure: A bull, a dart, a diamond, and a gem}.

\label{subsec:tech}
\begin{figure}
	\centering
	\begin{tikzpicture}[scale=0.25]
	\draw[gray, thick] (0,-{sqrt(3)}) -- (0,{sqrt(3)});
	\draw[gray, thick] (-2,0) -- (0,-{sqrt(3)});
	\draw[gray, thick] (-2,0) -- (0,{sqrt(3)});
	\draw[gray, thick] (2,0) -- (0,-{sqrt(3)});
	\draw[gray, thick] (2,0) -- (0,{sqrt(3)});
	\filldraw[black] (-2,0) circle (3pt);
	\filldraw[black] (2,0) circle (3pt);
	\filldraw[black] (0,-{sqrt(3)}) circle (3pt);
	\filldraw[black] (0,{sqrt(3)}) circle (3pt);
	\end{tikzpicture}
	\hspace{1cm}
	\begin{tikzpicture}[scale=0.25]
	\draw[gray, thick] (-3,1) -- (3,1);
	\draw[gray, thick] (-3,1) -- (0,-2);
	\draw[gray, thick] (-1,1) -- (0,-2);
	\draw[gray, thick] (1,1) -- (0,-2);
	\draw[gray, thick] (3,1) -- (0,-2);
	\filldraw[black] (-3,1) circle (3pt);
	\filldraw[black] (-1,1) circle (3pt);
	\filldraw[black] (1,1) circle (3pt);
	\filldraw[black] (3,1) circle (3pt);
	\filldraw[black] (0,-2) circle (3pt);
	\end{tikzpicture}
	\hspace{1cm}
	\begin{tikzpicture}[scale=0.45]
	\draw[gray, thick] (0,1) -- (0.951,0.309);
	\draw[gray, thick] (0,1) -- (-0.951,0.309);
	\draw[gray, thick] (0.951,0.309) -- (0.588,-0.809);
	\draw[gray, thick] (-0.951,0.309) -- (-0.588,-0.809);
	\draw[gray, thick] (0.951,0.309) -- (-0.951,0.309);
	\draw[gray, thick] (0.588,-0.809) -- (-0.588,-0.809);
	\filldraw[black] (0,1) circle (2pt);
	\filldraw[black] (0.951,0.309) circle (2pt);
	\filldraw[black] (-0.951,0.309) circle (2pt);
	\filldraw[black] (0.588,-0.809) circle (2pt);
	\filldraw[black] (-0.588,-0.809) circle (2pt);
	\end{tikzpicture}
	\hspace{1cm}
	\begin{tikzpicture}[scale=0.30]
	\draw[gray, thick] (-2,1) -- (2,1);
	\draw[gray, thick] (-2,-1) -- (2,-1);
	\draw[gray, thick] (-2,1) -- (-2,-1);
	\draw[gray, thick] (0,1) -- (0,-1);
	\draw[gray, thick] (2,1) -- (2,-1);
	\filldraw[black] (-2,1) circle (3pt);
	\filldraw[black] (-2,-1) circle (3pt);
	\filldraw[black] (0,1) circle (3pt);
	\filldraw[black] (0,-1) circle (3pt);
	\filldraw[black] (2,1) circle (3pt);
	\filldraw[black] (2,-1) circle (3pt);
	\end{tikzpicture}
	\hspace{1cm}
	\begin{tikzpicture}[scale=0.25]
	\draw[gray, thick] (-1,1) -- (1,1);
	\draw[gray, thick] (0,-{sqrt(3)}+1) -- (-1,1);
	\draw[gray, thick] (0,-{sqrt(3)}+1) -- (1,1);
	\draw[gray, thick] (-1,1) -- (-1,2.8);
	\draw[gray, thick] (1,1) -- (1,2.8);
	\filldraw[black] (-1,1) circle (3pt);
	\filldraw[black] (-1,2.8) circle (3pt);
	\filldraw[black] (1,1) circle (3pt);
	\filldraw[black] (1,2.8) circle (3pt);
	\filldraw[black] (0,-{sqrt(3)}+1) circle (3pt);
	\end{tikzpicture}
	\hspace{1cm}
	\begin{tikzpicture}[scale=0.25]
	\draw[gray, thick] (-1,0) -- (1,0);
	\draw[gray, thick] (0,-{sqrt(3)}) -- (-1,0);
	\draw[gray, thick] (0,-{sqrt(3)}) -- (1,0);
	\draw[gray, thick] (0,{sqrt(3)}) -- (-1,0);
	\draw[gray, thick] (0,{sqrt(3)}) -- (1,0);
	\draw[gray, thick] (-3,0) -- (-1,0);
	\filldraw[black] (-3,0) circle (3pt);
	\filldraw[black] (-1,0) circle (3pt);
	\filldraw[black] (1,0) circle (3pt);
	\filldraw[black] (0,-{sqrt(3)}) circle (3pt);
	\filldraw[black] (0,{sqrt(3)}) circle (3pt);
	\end{tikzpicture}
	\caption{A \dia, a \gem, a {\sf house}, a {\sf domino}, a \bull, and a \dart}
	\label{Figure: A bull, a dart, a diamond, and a gem}
\end{figure}
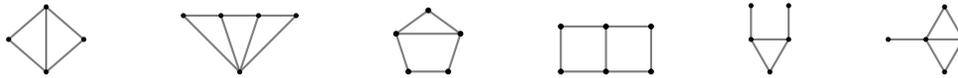

Chordal graphs are arguably the simplest graph class, apart from forests,  which is characterized by infinite forbidden induced subgraphs.
Structural and algorithmic aspects of chordal graphs have been extensively studied in the last decades, and it is considered one of the basic graph classes whose properties are well understood and on which otherwise NP-hard problems become tractable.
As such, it is natural to ask how close a graph to a chordal graph in terms of graph edit distance and there is a large body of literature  pursuing this topic~\cite{agrawal2018feedback,agrawal2018polylogarithmic,CaoM16,jansen2017approximation,KaplanST94,marx2010chordal,Yan81}.

Fixed-parameter tractability and the existence of polynomial kernel of \FD\ for chordal graphs were one of important open questions in parameterized complexity~\cite{marx2010chordal,jansen2017approximation}.
An affirmative answer to the latter in~\cite{jansen2017approximation}  brought the approximability for chordal graphs to the fore as it uses an $O(\mathsf{opt}^2 \log \mathsf{opt} \log n)$-factor approximation algorithm as a crucial subroutine.
It was soon improved to $O(\mathsf{opt}\log n)$-factor approximation~\cite{agrawal2018feedback,KK2018}.
An important step was taken by Agrawal et al.~\cite{agrawal2018polylogarithmic} who studied \WFD for chordal graphs, distance-hereditary graphs, and graphs of bounded treewidth.
They presented $\polylog(n)$-approximation algorithms for them, including $O(\log^2 n)$-approximation for chordal graphs, and left the existence of constant-factor approximation algorithms as an open question.
For now, even the existence of $O(\log n)$-factor approximation is not known.
This makes an interesting contrast with \FD\ for forests, that is, \textsc{Feedback Vertex Set}.
An algorithmic proof of Erd\"{o}s-P\'{o}sa property\footnote{Any graph has either a vertex-disjoint packing of $k+1$ cycles, or a feedback vertex set of size $O(k\log k)$.} for cycles immediately leads to an $O(\log n)$-factor approximation for \textsc{Feedback Vertex Set} while the known gap function of Erd\"{o}s-P\'{o}sa property for induced $C_{\geq 4}$ is not low enough to achieve such an approximation factor~\cite{KK2018}.

Distance-hereditary graphs, in which any induced subgraph preserves the distances among all vertex pairs, form another important subclass of perfect graphs.
It is supposedly the simplest dense graph class captured by a graph width parameter; distance-hereditary graphs are precisely the graphs of rankwidth 1~\cite{Oum05}.
\FD\ for distance-hereditary graphs has gained good attention for fixed-parameter tractability and approximability~\cite{agrawal2018polylogarithmic,kim2017polynomial,eiben2018single} particularly due to the recent surge of interest in rankwidth.
An $O(\log^3 n)$-approximation is known~\cite{agrawal2018polylogarithmic}.

Constant-factor approximation algorithms were designed for smaller subclasses of chordal and distance-hereditary graphs.
They include block graphs (excluding $C_{\geq 4}$ and a \dia)~\cite{agrawal2016faster} and 3-leaf power graphs (excluding $C_{\geq 4}$, a \bull, a \dart, and a \gem)~\cite{BL2006}.
See Figure~\ref{Figure: A bull, a dart, a diamond, and a gem}.
Recently, a $(2+\epsilon)$-factor approximation for split graphs was announced~\cite{LMPPS2020}.

In this paper, we take a step towards the (affirmative) answer of the question of~\cite{agrawal2018polylogarithmic} by presenting a constant-factor approximation algorithm for the intersection of chordal and distance-hereditary graphs, known as {\em \pto graphs.}\footnote{The name {\em \pto} comes from the fact that the shortest path distance satisfies {\em Ptolemy's inequality}: For every four vertices $u, v, w, x$, the inequality $d(u,v)d(w,x)+d(u,x)d(v,w)\geq d(u,w)d(v,x)$ holds.} 
They are precisely graphs without any induced $C_{\geq 4}$ or a \gem, so it is easy to see that they form a superclass of both 3-leaf power and block graphs.

\vskip 0.2cm
\noindent
\fbox{\parbox{0.97\textwidth}{
\WPD\\
\textbf{Input :} A graph $G = (V, E)$ with vertex weights $w : V \to \R^+ \cup \{ 0 \}$. \\
\textbf{Question :} Find a set $S \subseteq V$ of minimum weight such that $G \setminus S$ is \pto. }}
\vskip 0.2cm

\begin{THM}\label{main}
	\WPD admits a polynomial-time constant-factor approximation algorithm.
\end{THM}

\subsection{Techniques}
Our proof presents new properties and algorithmic results on {\em inter-clique} digraphs as well as an approximation algorithm for a variant of \FVS that exploits this relationship (named \FVSP), each of which may be of independent interest.

\subsubsection{Inter-clique Digraphs}

The starting point of our proof is to examine what we call an \emph{inter-clique digraph} of $G$.
Let $\mathcal{C}(G)$ be the collection of all non-empty intersections of maximal cliques in $G$, see Section~\ref{sec: preliminaries} for the formal definition.
An inter-clique digraph $\dT(G)$ of $G$, or simply $\dT$, is a digraph isomorphic to  the Hasse diagram of $(\mathcal{C}(G),\subseteq)$.
A neat characterization of ptolemaic graphs was presented by Uehara and Uno~\cite{UU2009}: a graph $G$ is ptolemaic if and only if its inter-clique digraph is a forest.
This immediately suggests the use of an $O(1)$-approximation algorithm for \FVS on the inter-clique digraph.
Indeed, the black-box application of an $O(1)$-approximation algorithm for \FVS yields $O(1)$-approximation algorithms for subclasses of ptolemaic graphs including block graphs~\cite{agrawal2016faster} and 3-leaf power graphs~\cite{AEKO2019}.

However, to leverage this characterization for \PD, two issues need to be addressed.
First, a polynomial-time algorithm to construct an inter-clique digraph of the input graph $G$ is needed, while the size of an inter-clique digraph can be exponentially large for general graphs.
Second, even with the inter-clique digraph of polynomial size at hand, the application of \FVS remains nontrivial since (1) after deletion of vertices, the structure of the inter-clique digraph may drastically change, and (2) feedback vertex sets for the inter-clique digraph must satisfy additional constraints that a deletion of a node $C \in \mathcal{C}(G)$ must imply the deletion of all nodes reachable from it (because they are subsets of $C$ in $G$).
Addressing each of these issues boils down to understanding the properties 
of an inter-clique digraph and elaborating the relationship between the input graph and its inter-clique digraph.

For general graphs, their inter-clique digraphs are acyclic digraphs in which each node can be precisely represented by all sources that have a directed path to the node.
It turns out that eliminating from $G$ all induced subgraphs isomorphic to $C_4$ and \gem is essential for tackling the aforementioned issues.
We show that any hole of $G$ indicates the existence of a cycle in $\und(\dT)$, and vice versa when $G$ is ($C_4$, \gem)-free (Lemmas~\ref{lem:four}-\ref{lem:nocommonanc}).
This in turn lets us to identify a variant of \WFVS, termed \FVSP and defined in Section~\ref{subsec:tech_fvsp}, which is essentially equivalent to \PD\ on $G$ when it takes the inter-clique digraph of $G$ as an input; see Proposition~\ref{prop:correctalgo}.
Moreover, each subdigraph of $\dT$ induced by the ancestors of any node $v$ of $\dT$ is a directed tree rooted at $v$, see Lemma~\ref{lem: laminar}.
(Similar statement holds for the descendants of $v$.)
This property is used importantly in analyzing our approximation for \FVSP.
As \FVSP\ takes an inter-clique digraph as an input, we need to construct it in polynomial time.
This is prohibitively time-consuming for general graphs.
We show that the construction becomes efficient when $G$ is both $C_4$ and \gem-free, see Proposition~\ref{thm:ptimeT}.

\subsubsection{Feedback Vertex Set with Precedence Constraints}\label{subsec:tech_fvsp}

Given acyclic directed graphs $\dG$ and a vertex $v$, let $\anc(v)$ and $\des(v)$ be the set of ancestors and descendants respectively, and let $\und(\dG)$ denote the underlying undirected graph of $\dG$.
It remains to design a constant-factor approximation algorithm for the following problem:

\vskip 0.2cm
\noindent
\fbox{\parbox{0.97\textwidth}{
\FVSP (\FVSPs) \\
\textbf{Input :} An acyclic directed graph $\dG = (V, A)$, where each vertex $v$ has weight $\omega_v \in \R^+ \cup \{ 0 \}$. 
For each $v \in V$, the subgraph induced by $\anc(v)$ is an in-tree rooted at $v$. 
\\
\textbf{Question :} Delete a minimum-weight vertex set $S \subseteq V$ such that (1) $v \in S$ implies $\des(v) \subseteq S$, (2) $\und(\dG \setminus S)$ is a forest. 
}}
\vskip 0.2cm

It is a variant of \textsc{Undirected Feedback Vertex Set} (\FVSs) on $\und(\dG)$, with the additional precedence constraint on $S$ captured by directions of arcs in $A$.
This precedence constraint makes an algorithm for \FVSPs harder to analyze than \FVSs because a vertex $v$ can be deleted ``indirectly''; even when $v$ does not participate in any cycle, deletion of any ancestor of $v$ forces to $v$ to be deleted, so the analysis for $v$ needs to keep track of every vertex in $\anc(v)$.

We adapt a recent constant-factor approximation algorithm for \SFVS by Chekuri and Madan~\cite{chekuri2016constant} for \FVSPs.
The linear programming (LP) relaxation variables are $\{ z_v \}_{v \in V}$, where $z_v$ is supposed to indicate whether $v$ is deleted or not, as well as $\{ x_{ue} \}_{e \in A, u \in e}$, where $x_{ue}$ is supposed to indicate that in the resulting forest $\und(\dG \setminus S)$ rooted at arbitrary vertices, whether $e$ is the edge connecting $u$ and its parent.

\vspace{-8pt}
\begin{align}
\mbox{Minimize }&  \sum_{v \in V} z_v \omega_v \nonumber \\
\mbox{Subject to }
& z_v + x_{ue} + x_{ve} = 1 \mbox{ for each } e = (u, v) \in A, 
&& z_v + \sum_{e \ni v} x_{ve} \leq 1 \mbox{ for each } v \in V,  \nonumber \\
& z_u \leq z_v \mbox{ for each } e = (u, v) \in A, 
&& 0 \leq x, z \leq 1. \nonumber
\end{align}

Compared to the LP in~\cite{chekuri2016constant}, we added the $z_u \leq z_v$ for all $(u, v) \in A$ to encode the fact that $u$'s deletion implies $v$'s deletion.
This LP is not technically a relaxation, but one can easily observe that in any integral solution, the graph induced by $\{ v : z_v = 0\}$ has at most one cycle, which can be easily handled later.\footnote{\cite{chekuri2016constant} added an additional cycle covering constraint in the LP. We find it conceptually easier to deal with the last remaining cycle separately at the end. }
The rounding algorithm proceeds as follows.
Fix three parameters $\eps \approx 0.029, \alpha \approx 0.514, \beta \approx 0.588$.
For notational convenience, let $\bx_{ue} := 1 - x_{ue}$.
Also, for each $e = (u, v) \in A$, let $y_{e} = z_v - z_u$.

\begin{enumerate}[(i)]
	\item Delete all vertex $v$ with $z_v \geq \eps$.
	\item Sample $\theta$ uniformly at random from the interval $[\alpha, \beta]$.
	\item For each $e = (u, v) \in A$, if $\theta \in [\bx_{ve} - y_{e}, \bx_{ve}]$, delete $\des(v)$.
\end{enumerate}

Slightly modifying the analysis of~\cite{chekuri2016constant}, one can show that after rounding, there is indeed at most one cycle remained in each connected component.
In terms of the total weight of deleted vertices, it is easy to bound the total weight of deleted vertices in Step (i) and the final cleanup step for one cycle.
The main technical lemma of the analysis bounds the weight of vertices deleted in Step (iii) by at most $O(\lp)$.

\begin{LEM}
	For each $v \in V$, $\Pr[v\mbox{ is deleted in Step (iii)}] \leq O(z_v)$.
\end{LEM}

Recall that $\anc(v)$ induces the directed tree $\dT$ rooted on $v$ where all arcs are directed towards $v$, and deletion of any vertex in $\dT$ forces the deletion of $v$.
The lemma is proved by showing that while $\anc(v)$ can be large, all vertices that can be possibly deleted during the rounding algorithm can be covered by at most two directed paths; it is proved by examining behaviors of the rounding algorithm on directed trees, followed by an application of Dilworth's theorem.
The new LP constraint $z_u \leq z_v$ for all $(u, v) \in A$ ensures that the sum of the deletion probabilities along any path is at most $O(z_v)$, so the total probability that $v$ is deleted can be bounded by $O(z_v)$.

\section{Preliminaries}\label{sec: preliminaries}

For a mapping $f:X\rightarrow Y$ between two finite sets and a set $A\subseteq X$, we denote $\bigcup_{x\in A}f(x)$ by $f(A)$.
For sets $X$ and $Y$, we say that $X$ and $Y$ are \emph{overlapping} if none of $X\setminus Y$, $Y\setminus X$, and $X\cap Y$ is empty.
For a family $\mathcal{F}$ of sets, $\mathcal{F}$ is \emph{laminar} if $\mathcal{F}$ has no overlapping two elements.

\medskip

\noindent {\bf Graph terminology.}
In this paper, all (directed) graphs are finite and simple.

Let $G=(V,E)$ be an undirected graph.
We often write the vertex set of $G$ as $V(G)$ and its edge set as $E(G)$.
For a vertex $v$ of $G$ and subsets $X$ and $Y$ of $V(G)$, let $N_G(v)$ be the set of neighbors of $v$ in $V(G)$, and $N_G(X)$ be the set of vertices not in $X$ that are adjacent to some vertices in $X$.
When the graph under consideration is clear in the context, we omit the subscript.
For two disjoint vertex sets $X$ and $Y$ of $G$, we say that $X$ is \emph{complete} to $Y$ if $x$ and $y$ are adjacent in $G$ for every $x\in X$ and $y\in Y$.
We say that two vertices $u,v$ are \emph{true twins}, or simply twins, if $N_G[u]=N_G[v]$.
Note that true twins must be adjacent.
Since the true twin relation is an equivalent relation, the true twin classes of $V$ is uniquely defined.

Let $\dG:=(V,A)$ be a directed graph.
The vertex set of $\dG$ is sometime written as $V(\dG)$, and its arc set as $A(\dG)$.
We denote by $\und(\dG)$ the underlying graph of $\dG$.

A \emph{source} of $\dG$ is a vertex of $\dG$ without an in-coming arc and a \emph{sink} of $\dG$ is a vertex without an out-going arc.
We say that $v$ is \emph{reachable} from $u$ in $\dG$ if $\dG$ has a directed path of length from $u$ to $v$.
An \emph{ancestor} of $v$ in $\dG$ is a vertex which is reachable to $v$ in $\dG$ and a \emph{descendant} of $v$ in $\dG$ is a vertex which is reachable from $v$ in $\dG$.
Two vertices $u$ and $v$ are \emph{incomparable} in $\dG$ if neither one is an ancestor of the other.
For distinct vertices $v_1,\ldots,v_\ell$ of $\dG$ with $\ell\geq2$, a \emph{least common ancestor} of $v_1,\ldots,v_\ell$ in $\dG$ is a common ancestor $w$ of $v_1,\ldots,v_\ell$ in $\dG$ such that a descendant $u$ of $w$ in $\dG$ is a common ancestor of $v_1,\ldots,v_\ell$ in $\dG$ if and only if $u=w$.
Similarly, a \emph{greatest common descendant} of $v_1,\ldots,v_\ell$ in $\dG$ is a common descendant $w$ of $v_1,\ldots,v_\ell$ in $\dG$ such that an ancestor $u$ of $w$ in $\dG$ is a common descendant of $v_1,\ldots,v_\ell$ in $\dG$ if and only if $u=w$.
Let $\anc(\dG,v)$ be the set of ancestors of $v$ in $\dG$, $\des(\dG,v)$ be the set of descendants of $v$ in $\dG$, and $\rt(\dG,v)$ be the set of sources of $\dG$ which are ancestors of $v$ in $\dG$.
When $\dG$ is clear from the context, we may simply write $\anc(v)$, $\des(v)$, and $\rt(v)$, respectively.
We say that $\dG$ is an \emph{out-tree} (respectively, \emph{in-tree}) if $\dG$ has a unique source (respectively, sink) $r\in V$, called the \emph{root}, and 
every arc is oriented away from (respectively, toward) $r$.

For a cycle $H$ in $\dG$ which is not a directed cycle, we term a maximal directed subpath of $G$ a \emph{segment} of the cycle $H$.
It is clear that the number of segments of $H$ is even (and non-zero) when $H$ is not a directed cycle.
The \emph{segment length} of a cycle $H$ is defined as the number of segments of $H$.
A \emph{segment decomposition} of a cycle $H$ is a cyclic sequence of all segments of $H$ such that any two consecutive segments share a vertex of $\dG$.
We will write a segment decomposition of $H$  as $H=x_0,\vec{P_1}, x_1, \cev{P_2},x_2, \cdots ,x_{2\ell-1}, \cev{P_{2\ell}}, x_{2\ell}(=x_0)$, in which for every odd $i$, $\vec{P}_i$ is a forward-oriented path from $x_{i-1}$ to $x_i$ and for every even $i$, $P_i$ is a backward-oriented path from $x_{i-1}$ to $x_{i}$ (addition is taken modulo $2\ell$, i.e., the segment length of $H$).
To emphasize the orientation of each path $P_i$, we write $\vec{P_i}$ for every odd $i$ and $\cev{P_i}$ for every even $i$.
We use a segment decomposition with the minimum number of segments; in such a decomposition, the number of segments is always even.

For a (directed) graph $G$ and a set $X\subseteq V(G)$, let $G\setminus X$ be a (directed) graph obtained from $G$ by removing all vertices in $X$ and all edges or arcs incident with some vertices in $X$, and $G[X]:=G\setminus(V(G)\setminus X)$.
We may write $G\setminus v$ instead of $G\setminus\left\{v\right\}$.
For an undirected graph $G$ and a set $Y\subseteq E(G)$, let $G/Y$ be a graph obtained from $G$ by contracting all edges in $Y$.

\medskip

\noindent {\bf Clique and inter-clique digraph.}
A \emph{clique} of $G$ is a set of pairwise adjacent vertices of $G$.
We denote the set of maximal cliques in a graph $G$ by $\M(G)$.
We define the set $\C(G)$ all non-empty intersections among maximal cliques, that is,
\begin{equation*}
	\C(G):=\bigcup_{\I \subseteq\M(G)}\left\{C:C=\bigcap_{M\in \I}M,\ C\neq\emptyset\right\}.
\end{equation*}
When the reference graph $G$ is clear in the context, we write $\M(G)$ and $\C(G)$ as $\M$ and $\C$ respectively.

Cleary, $\C(G)$ defines a partially ordered set under the set containment relation $\subseteq$.
A \emph{Hasse diagram} $\dH$ of a poset $(S,\leq)$ represents each element of $S$ as a vertex and adds an arc from $y$ to $x$ if and only if $y>x$ and there is no element $z\in S$ with $y>z>x$.
We say that a digraph $\dT$ is an \emph{inter-clique digraph} of $G$ if $\dT$ isomorphic to the Hasse diagram of the poset $(\C(G),\subseteq)$.
For an inter-clique digraph $\dT$ of $G$ or the Hasse diagram $\dH$, we call $V(\dT)$ or $V(\dH)$ \emph{nodes} instead of vertices in order to distinguish them from the vertices of $G$.

For a vertex set $X\subseteq V(G)$, we define $\rt(X)$ as the set of all maximal cliques containing $X$.
In  case $X$ is a singleton consisting of $v$, we omit the bracket and write $\rt(v)$ instead of $\rt(\{v\})$.
For a collection of sets $\mathcal{X}$, $\rt(\mathcal{X})$ is defined as the collection of sets (without duplicates) $\rt(\mathcal{X})=\{\rt(X):X\in \mathcal{X}\}$.
Clearly, a vertex set $X$ is a clique if and only if   $\rt(X)\neq \emptyset$.
The following observation is immediate from the fact that a clique is an ancestor of another clique in $\dH$ if and only if the former contain the latter.

\begin{OBS}\label{obs:defscoincide}
	Let $\dH$ be the Hasse diagram of $(\C(G),\subseteq)$.
	For a clique $C\in \C(G)$, we have $\rt(C)=\rt(\dH,C)$.
\end{OBS}

Observation~\ref{obs:defscoincide} justifies the reuse of the notation $\rt$ for a vertex set, while $\rt(\dG,v)$ is already defined to delineate the set of vertices with no in-coming arcs from which there is a directed path to $v$ in $\dG$.

\medskip

\noindent {\bf Ptolemaic graphs.}
For vertices $u$ and $v$ of $G$ in the same component, the distance between $u$ and $v$ in $G$, denoted by $\text{dist}_G(u,v)$, is the length of shortest path from $u$ to $v$.
A graph $G$ is \emph{distance-hereditary} if for every connected induced subgraph $H$ of $G$ and vertices $v$ and $w$ of $H$, $\text{dist}_H(u,v)=\text{dist}_G(u,v)$.
A graph is \emph{chordal} if it contains no hole, e.g., no induced cycle of length at least $4$.
For graphs $G_1,\ldots,G_m$, we say that a graph $G$ is \emph{$(G_1,\ldots,G_m)$-free} if $G$ has no induced subgraph isomorphic to one of $G_1,\ldots,G_m$.
A graph is \emph{ptolemaic} if for every four vertices $a$, $b$, $c$, and $d$ in the same component, $G$ satisfies the following inequality:
\begin{equation*}
	\text{dist}_G(a,b)\cdot\text{dist}_G(c,d)\leq\text{dist}_G(a,c)\cdot\text{dist}_G(b,d)+\text{dist}_G(a,d)\cdot\text{dist}_G(b,c).
\end{equation*}

Howorka~\cite{Howorka1981} presented characterizations of ptolemaic graphs.

\begin{THM}[Howorka~\cite{Howorka1981}]\label{Howorka1981}
	The following four conditions are equivalent.
	\begin{enumerate}
		\item [(1)] A graph $G$ is ptolemaic.
		\item [(2)] $G$ is distance-hereditary and chordal.
		\item [(3)] $G$ is \gem-free and chordal.
		\item [(4)] For every pair of distinct non-disjoint maximal cliques $M$ and $N$, $M\cap N$ separates $M\setminus N$ and $N\setminus M$, that is, every path in $G$ between a vertex in $M$ and $N$ must intersect a vertex in $M\cap N$.
	\end{enumerate}
\end{THM}

Uehara and Uno~\cite{UU2009} presented another characterization by showing that the maximal cliques in a ptolemaic graph represent a tree structure for the ptolemaic graph.

\begin{THM}[Uehara and Uno~\cite{UU2009}]\label{UU2009}
	A graph $G$ is ptolemaic if and only if $\und(\dH)$ is a forest, where $\dH$ is the Hasse diagram of $(\C(G),\subseteq)$.
\end{THM}

\section{Structures of Inter-clique digraphs}\label{sec: inter-clique digraphs}

\subsection{Basic properties of inter-clique digraphs} 

In this subsection, we investigate the properties of the Hasse diagram $\dH$ of the poset $(\C(G),\subseteq)$ for a graph $G=(V,E)$.
All the results presented in this subsection assume no restriction on the input graph $G$.

Recall that for a vertex set $X$ of $G$, $\rt(X)\neq \emptyset$ if and only if $X$ is a clique.
Our first observation is that $\C(G)$ consists precisely of those maximal cliques $X$ such that $\rt(X)$ remains unchanged.
It also provides a way to find the maximal cliques in $\M(G)$ whose intersection is equal to $C$.

\begin{LEM}\label{lem: intersection of maximal cliques}
	For a clique $C$ of $G$, we have $C\in \C(G)$ if and only if $C=\bigcap_{M\in \rt(C)} M$.
\end{LEM}
\begin{proof}
	The opposite direction is immediate from the definition of $\C(G)$.
	To see the forward direction, let $\I\subseteq \M(G)$ be a maximal set such that $C=\bigcap_{M\in \I}M$ and notice that $C$ is contained in each maximal clique of $\I$.
	Therefore, we have $\I\subseteq \rt(C)$, and equality holds due to the maximality of $\I$.
\end{proof}

The next lemma observes that each vertex $v$ of $V$ can be uniquely associated to a clique $C$ of $\C(G)$ with the 
property $\rt(v)=\rt(C)$.

\begin{LEM}\label{lem: labeling}
	For every vertex $v$ of $G$, there is a unique minimal element $C(v)\in\C(G)$ containing $v$ in the poset $(\C(G),\subseteq)$ and 
	it holds that $C(v)=\bigcap_{M\in \rt(v)}M$.
\end{LEM}
\begin{proof}
	Note that every vertex of $G$ is contained in at least one maximal clique in $G$.
	Suppose there are two minimal element $C, C'\in \C(G)$ in the poset containing $v$.
	Thus, $C'\nsubseteq C$ and $C\nsubseteq C'$, and therefore $C\cap C'$ is a non-empty proper subset of $C$.
	Moreover,  observe that $v\in C\cap C'= \bigcap_{M\in \rt(C)\cup \rt(C')} M \in\C(G)$, where the equality holds due to Lemma~\ref{lem: intersection of maximal cliques}.
	This contradicts the minimality of $C$, thus establishing the uniqueness of a minimal $C(v)$ containing $v$.
	
	To see the second statement, consider the set of all clique of $\C(G)$ containing $v$.
	Due to the uniqueness of a minimal element $C(v)$ containing $v$, it holds that any clique $C\in \C(G)$ contains $v$ if and only $C\supseteq C(v)$.
	In particular, this implies $\rt(v)=\rt(C(v))$, and together with Lemma~\ref{lem: intersection of maximal cliques} the second statement follows.
\end{proof}

We call the clique as depicted in Lemma~\ref{lem: labeling} the \emph{canonical clique} of $v$, namely the canonical clique is defined as $C(v)=\bigcap_{M\in \rt(v)}M$.
Note that $\rt(v)=\rt(C(v))$.

\begin{LEM}\label{lem:computetwin}
	Let $u$ and $v$ be two adjacent vertices of $G$.
	The followings are equivalent.
	\begin{enumerate}[(i)]
		\item The canonical  cliques of $u$ and $v$ are identical, i.e., $C(u)=C(v)$.
		\item $u$ and $v$ are (true) twins in $G$.
		\item A maximal clique contains  $u$ if and only if it contains $v$, i.e., $\rt(u)=\rt(v)$.
	\end{enumerate} 
\end{LEM}
\begin{proof}
	To see that (i) implies (ii), let $w$ be an arbitrary neighbor of $u$.
	Note that a maximal clique $M$ containing the edge $uw$ contains $C(u)$ as well by Lemma~\ref{lem: labeling}.
	It follows that $v\in C(v)=C(u)\subseteq M$, and thus $w$ is a neighbor of $v$ as well.
	Suppose (iii) does not hold, and without loss of generality let $M$ be a maximal clique in $\rt(u)\setminus \rt(v)$.
	Then there exists a vertex in $M$ which is not adjacent with $v$ since otherwise $M\cup \{v\}$ is a clique, contradicting the maximality of $M$.
	This means $u$ and $v$ are not true twins, thus establishing the implication from (ii) to (iii).
	That (iii) implies (i) follows from Lemma~\ref{lem: labeling}, which asserts $C(v)=\bigcap_{M\in \rt(v)} M = \bigcap_{M\in \rt(u)} M=C(u)$.
\end{proof}

The next lemma offers how to read off the relation between two nodes of $\dH$ from the mapping $\rt$.
Essentially, it says that $\dH$ is the reversal of the Hasse diagram of $(\rt(\C(G)),\subseteq)$, where $\rt(\C(G))=\{\rt(C):C\in \C(G)\}$.
We shall use this lemma extensively in the later proofs, and may sometimes omit to refer to it.

\begin{LEM}\label{lem:labelrelation}
	Let $C,C'$ be two cliques of $\C(G)$.
	Then there is a directed path from $C$ to $C'$ in $\dH$ if and only if $\rt(C) \subseteq \rt(C')$, where the equality holds only if $C=C'$.
\end{LEM}
\begin{proof}
	If $C=C'$, that $\rt(C)=\rt(C')$ is obviously.
	Conversely, $\rt(C)=\rt(C')$ implies $C=C'$ by Lemma~\ref{lem: intersection of maximal cliques} and that is, the equality holds only if $C=C'$.
	Therefore, we may assume that $C\neq C'$.
	Since $C'$ is reachable from $C$, $C'$ is reachable by all maximal cliques of $\rt(C)$ and thus $\rt(C')$ is a superset of $\rt(C)$.
	Conversely, if $\rt(C) \subsetneq \rt(C')$, we have $C'=\bigcap_{M\in \rt(C')} \subseteq \bigcap_{M\in \rt(C)}=C$ by Lemma~\ref{lem: intersection of maximal cliques}, and especially $C'\subsetneq C$.
	It follows that $C'$ is reachable from $C$ in $\dH$.
\end{proof}

The next lemma observes that even when a node $C$ of $\dH$ has many immediate descendants, we can fully describe $C$ by considering two arbitrary immediate descendants of $C$.

\begin{LEM}\label{lem:twoenough}
	If a node $C$ has immediate descendants $C_1,\ldots,C_p$ with $p\geq 2$ in $\dH$, then we have $\rt(C)=\rt(C_i)\cap \rt(C_j)$ for every $1\leq i<j\leq p$.
\end{LEM}
\begin{proof}
	Suppose not, i.e., we have $\rt(C)\neq \rt(C_1)\cap \rt(C_2)$ without loss of generality.
	Because it holds that $\rt(C)\subsetneq \rt(C_i)$ for $i\in [2]$ by Lemma~\ref{lem:labelrelation}, this means that $\rt(C)\subsetneq \rt(C_1)\cap \rt(C_2)$.
	Observe that the clique $C'=\bigcup_{M\in \rt(C_1)\cap \rt(C_2)} M$ contain both $C_1$ and $C_2$ by Lemma~\ref{lem:labelrelation}, and thus $C_1\cup C_2$ (possibly some more vertices).
	Hence, $C'$ is non-empty.
	In particular, $C'$ is a member of $\C(G)$ and there is a directed path in $\dH$ from $C'$ to $C_i$ for $i\in [2]$.
	Now, the relation $\rt(C)\subsetneq \rt(C_1)\cap \rt(C_2)$ implies $C\supsetneq C' \supsetneq C_i$ for $i\in [2]$ by Lemma~\ref{lem:labelrelation}.
	This contradicts that  there is an arc from $C$ to $C_i$ for $i\in[2]$ in $\dH$.
\end{proof}

\begin{LEM}\label{lem:lastnode4twin}
	Let $Z$  be a true twin class of $G$ contained in a clique $C\in \C(G)$.
	Then the following are equivalent.
	\begin{enumerate}[(i)]
		\item $\rt(C)\subsetneq \rt(Z)$.
		\item There exists a proper descendant $C'$ of $C$ in $\dH$ such that $Z\subseteq C'$.
	\end{enumerate}
\end{LEM}
\begin{proof}
	$(i)\rightarrow (ii)$: Note that $Z\neq \emptyset$, and Lemma~\ref{lem:computetwin} subsumes $Z=\bigcap_{M\in \rt(Z)} M$.
	Hence $Z$ is a clique of $\C(G)$.
	From $\rt(C)\subsetneq \rt(Z)$, we know that $Z$ is a proper descendant of $C$ by Lemma~\ref{lem:labelrelation}.

	\smallskip

	\noindent $(ii)\rightarrow (i)$: Suppose it does not hold that $\rt(C)\subsetneq \rt(Z)$.
	For every vertex of $C$ is contained in each maximal clique of $\rt(C)$, we have $\rt(C)\subseteq \rt(Z)$, and thus $\rt(C)= \rt(Z)$.
	Choose a descendant $C'$ of $C$ in $\dH$ containing $Z$ and observe that we have $\rt(C)\subsetneq \rt(C')$ by Lemma~\ref{lem:labelrelation}.
	Therefore, there exists a maximal clique $M^{\star}\in \rt(C')\setminus \rt(C)=\rt(C')\setminus \rt(Z)$ which does not contain $Z$ entirely.
	This means that $Z$ is not contained in $C'$, a contradiction.
\end{proof}

The next few lemmas interpret some obvious properties of an intersection of maximal cliques in the Hasse diagram setting: if $C$ is an intersection of maximal cliques $\I\subseteq \M(G)$, then $C$ is the unique minimal clique containing all cliques contained in $C$ and it is also the unique maximal clique contained in all cliques containing $C$.

\begin{LEM}\label{lem:labelmaximal}
	For $\I\subseteq \M(G)$, let $\C'\subseteq \C(G)$ be the set of all cliques $C$ of $\C(G)$ such that $\I\subseteq \rt(C)$.
	Then there exists at most one maximal element in $\C'$.
\end{LEM}
\begin{proof}
	We may assume $\C'\neq \emptyset$ since otherwise the statement trivially holds.
	For the sake of contradiction, suppose that $C_1,C_2,\ldots , C_s\in \C'$ are the maximal elements of $\C'$ with $s\geq 2$.
	Note that none of $\rt(C_i)$ contains  $\rt(C_j)$ for $i\neq j$ due to the maximality assumption of $C_1,\ldots , C_s$ and Lemma~\ref{lem:labelrelation}.
	Now let $\I'=\bigcap_{i\in [s]} \rt(C_i)$ and notice that $\I'\subsetneq \rt(C_i)$ for every $i\in [s]$ due to the previous argument.
	Now, for every $i\in [s]$: 
	\[
	C_i=\bigcap_{M\in \rt(C_i)}M = \bigcap_{M\in \I'}M \cap \bigcap_{M\in \rt(C_i)\setminus \I'}M \subseteq \bigcap_{M\in \I'}M.
	\]
	Therefore, $C^*=\bigcap_{M\in \I'}M$ is not only a non-empty clique, but also contains every $C_i$.
	Finally we observe that $\I \subseteq \I' \subseteq \rt(C^*)$, and thus $C^*\in \C'$.
	This contradicts the choice of $C_1,\ldots , C_s$ as maximal elements of $\C'$.
\end{proof}

\begin{LEM}\label{lem: at most 1 common}
	Let $C_1$ and $C_2$ be two cliques of $\C(G)$.
	Then $\dH$ contains at most one greatest common descendant of $C_1$ and $C_2$.
\end{LEM}
\begin{proof}
	Observe that $C\in \C(G)$ is a common descendant of $C_1$ and $C_2$ if and only if $\rt(C)\supseteq \rt(C_1)\cup \rt(C_2)$ by Lemma~\ref{lem:labelrelation}.
	Now applying Lemma~\ref{lem:labelmaximal} with $\I=\rt(C_1)\cup \rt(C_2)$ proves the statement.
\end{proof}

\subsection{Inter-clique digraphs of ($C_4$, \gem)-free graphs}\label{subsec: modifications}

In the later subsections, we demonstrate how to forge an approximate solution to \PD\ using an approximation algorithm for \FVSP.
To this end, we examine how the extra assumption that $G$ is ($C_4$, \gem)-free brings about a new structure to emerge in the corresponding Hasse diagram.
Unless stated otherwise explicitly, $\dH$ refers to the Hasse diagram of $(\C(G),\subseteq)$ for a ($C_4$, \gem)-free graph $G=(V,E)$.

\begin{LEM}\label{lem: laminar}
	Let $G=(V,E)$ be a ($C_4$, \gem)-free graph and $M$ be a maximal clique of $G$.
	Then $\C_M:=\{C\in \C(G):C\subseteq M\}$ is laminar and $\dH[\C_M]$ is an out-tree rooted at $M$.
\end{LEM}
\begin{proof}
	Suppose that $\C_M$ contains overlapping elements $C_1$ and $C_2$.
	Note that none of $C_1$ and $C_2$ is $M$.
	Let $c_1$ be an element in $C_1\setminus C_2$, $c_2$ be an element in $C_2\setminus C_1$, and $c$ be an element in $C_1\cap C_2$.
	By the construction of $\C(G)$, there are maximal cliques $M_1$ and $M_2$ such that $C_1\subseteq M_1$, $C_2\subseteq M_2$, $c_2\notin M_1$, and $c_1\notin M_2$.
	Then $M_1\setminus(M\cup M_2)$ is non-empty, because otherwise $M_1$ is a proper subset of $C_1\cup C_2\cup(M_1\cap M_2)$ which is a clique in $G$.
	Similarly, $M_2\setminus(M\cup M_1)$ is non-empty.
	Let $m_1$ be an element in $M_1\setminus(M\cup M_2)$ and $m_2$ be an element in $M_2\setminus(M\cup M_1)$.
	Since every vertex in $M_1\cap(C_1\cup M_2)$ is adjacent to $c_2$, we may assume that $m_1$ is non-adjacent to $c_2$, because otherwise $M_1\cup\{c_2\}$ is a clique in $G$.
	Similarly, we may assume that $m_2$ is non-adjacent to $c_1$.
	Then $G[\{c,c_1,c_2,m_1,m_2\}]$ has a hole of length $4$ if $m_1$ and $m_2$ are adjacent, and is isomorphic to the \gem if $m_1$ and $m_2$ are non-adjacent, a contradiction.
	Therefore, $\C_M$ is laminar.
	
	To see that $\dH[\C_M]$ is an out-tree, we first note that $M$ is the unique maximal element  in $\dH[\C_M]$ by Lemma~\ref{lem:labelmaximal}.
	Therefore, it suffices to prove that $M$ has a unique path to any node $C\in \C_M$ in $\dH$.
	Suppose not, which means there exists $C\in \C_M$ and two vertex-disjoint paths from $M$ to $C$ in $\dH$.
	Let $C_1$ and $C_2$ be the immediate ancestor of $C$ on these two paths.
	Since $C,C_1$ and $C_2$ are all distinct cliques and $C\subseteq C_i$ for $i=1,2$, both $C_1\setminus C$ and $C_2\setminus C$ are non-empty.
	
	We argue that $C_1\setminus C$ and $C_2\setminus C$ are disjoint.
	Indeed, if a vertex $v$ of $G$ belongs to both $C_1\setminus C$ and $C_2\setminus C$, then $C(v)$ is a common descendant of both $C_1$ and $C_2$ as it is the unique minimal element of all elements of $\C(G)$ containing $v$ by Lemma~\ref{lem: labeling}.
	On the other hand, Lemma~\ref{lem: at most 1 common} implies that $C$ is the (unique) greatest common descendant of $C_1$ and $C_2$.
	Therefore, $C(v)$ is a descendant of $C$.
	This means that $v\in C(v)\subseteq C$, contradicting the choice of $v$.
	
	Therefore, $C_1\setminus C$ and $C_2\setminus C$ are disjoint, which means $C_1$ and $C_2$ are overlapping.
	This contradicts the laminarity of $\C_M$, thus establishing that $\dH[\C_M]$ is an out-tree.
\end{proof}

\begin{LEM}\label{lem: morelaminar}
	Let $G=(V,E)$ be a ($C_4$, \gem)-free graph, let $\I\subseteq \M(G)$, and let $\C'\subseteq \C(G)$ be the set of all elements $C$ such that $\I\subseteq \rt(C)$.
	Then the subdigraph of $\dH$ induced by $\C'$ is an out-tree.
	Consequently, if $C$ and $C'$ are elements in $\C(G)$, then $\dH$ contains at most one directed path from $C$ to $C'$.
\end{LEM}
\begin{proof}
	If $\C'=\emptyset$, there is nothing to prove.
	Otherwise, Lemma~\ref{lem:labelmaximal} subsumes that there is a unique maximal element $C$ in $\C'$.
	Consider an arbitrary maximal clique $M$ which is an ancestor of $C$, possibly $C=M$.
	Now, the uniqueness and the maximality of $C$ in $\C'$ implies that every element of $\C'$ is a descendant of $M$, and thus $\C'$ induces an out-tree in $\dH$ by Lemma~\ref{lem: laminar}.

	To see the second statement, assume that $C'$ is reachable from $C$; otherwise, the statement is vacuously valid.
	Lemma~\ref{lem:labelrelation} implies that there is a directed path from $C$ to $C'$ if and only if $\rt(C)$ is a subset of $\rt(C')$.
	Now, applying the first statement with $\I=\rt(C)$ yields the statement immediately.
\end{proof} 

To prove Lemma~\ref{lem: hereditary}, we need the following lemma which was proved in~\cite{KK2018}.

\begin{LEM}\label{lem:twopaths}\cite{KK2018}
	Let $G$ be a graph, $P$ and $Q$ be internally vertex-disjoint $(x,y)$-paths and let $w\in V(Q)\setminus \{x,y\}$ have no neighbor in $V(P)\setminus \{x,y\}$.
	If $Q$ is an induced path, then $G[P\cup Q]$ contains a hole.
\end{LEM}

\begin{LEM}\label{lem: hereditary}
	Let $G=(V,E)$ be a ($C_4$, \gem)-free graph.
	If $G$ has a hole $H$ and $v\in V(H)$, then $G[V(H)\cup \{v'\}\setminus \{v\}]$ contains a hole for every $v'\in C(v)$.
\end{LEM}
\begin{proof} 
	If $C(v)=\{v\}$, then the statement holds trivially.
	So we assume $\abs{C(v)}\geq 2$.

	First, we claim that $V(H)\cap C(v)=\{v\}$.
	Suppose that $V(H)\cap C(v)$ contains an element $u$ different from $v$.
	Note that $u$ is adjacent to $v$ because $C(v)$ is a clique.
	Moreover, for an arbitrary maximal clique $M\in \rt(v)$, Lemma~\ref{lem: labeling} implies that $M$ contains the clique $C(v)$.
	This in particular implies that any neighbor of $v$ is a neighbor of $u$ as well.
	Then, $\{u,v,w\}$ forms a triangle in $G$, where $w$ is the neighbor of $v$ on $H$ different from $u$, contradicting to the assumption that $u,v,w$ lie on the hole $H$.
	Therefore, we have $V(H)\cap C(v)=\{v\}$.
	
	Secondly, an arbitrary vertex $v'\in C(v)\setminus v$ and let $u,w,x\in V(H)$ be vertices such that $G[\{u,v,w,x\}]$ forms an induced subpath of $H$, where $u$ and $w$ are neighbors of $v$.
	Notice that the latter is possible as $G$ is $C_4$-free.
	We argue that $v'$ neighbors both $u$ and $w$ while it is non-adjacent with $x$.
	Obviously, there exist two maximal cliques $M_u$ and $M_w$ such that $\{u,v\}\subseteq M_u$ and $\{v,w\} \subseteq M_w$.
	By Lemma~\ref{lem: labeling}, both $M_u$ and $M_w$ contain the clique $C(v)$, and thus contain $v'$.
	This means that both $u$ and $w$ are adjacent with $v'$.
	If $v'$ is furthermore adjacent with $x$, then $G[\{u,v,w,x,v'\}$ induces a \gem, contradicting the assumption that $G$ is \gem-free.
	Therefore $v'$ neighbors both $u$ and $w$ while it is non-adjacent with $x$.
	
	Now, we are ready to apply Lemma~\ref{lem:twopaths}.
	Let $P$ be the subpath of $H$ between $u$ and $x$ avoiding $w$, and $Q=u,v',w,x$.
	The two paths $P$ and $Q$ are vertex-disjoint, especially because $v'$ is not included in $H$ due to $V(H)\cap C(v)=\{v\}$.
	Moreover, $w\in V(Q)$ has no neighbor in $V(P)\setminus \{u,x\}$.
	Clearly, $Q$ is an induced path by the argument of the previous paragraph, which implies that $G[V(P)\cup V(Q)]=G[V(H)\cup \{v'\}\setminus \{v\}]$ contains a hole by Lemma~\ref{lem:twopaths}.
\end{proof}

Recall that a graph $G$ is ptolemaic if and only if $\und(\dH)$ is a forest, where $\dH$ is the Hasse diagram of $(\C(G),\subseteq)$ (see Theorem~\ref{UU2009}).
Therefore, the Hasse diagram $\dH$ of $(\C(G),\subseteq)$ may still contain cycles when $G$ is a ($C_4$, \gem)-free graph.
In the rest of this subsection, we investigate the properties of cycles in $\dH$.
Due to the transitivity of the poset $(\C(G),\subseteq)$, there is no directed cycle in $\dH$ and the segment length of any cycle $H$ is even and at least two.
The next lemma states that the segment length is at least 8 when $G$ is ($C_4$, \gem)-free.

\begin{LEM}\label{lem:four}
	Let $G=(V,E)$ be a ($C_4$, \gem)-free graph.
	Then any undirected cycle $H$ of $\dH$ has segment length at least 8.
\end{LEM}
\begin{proof}
	Let $H=C_0,\vec{P_1}, C_1, \cev{P_2},C_2, \cdots ,C_{2\ell-1}, \cev{P_{2\ell}}, C_{2\ell}(=C_0)$ be a segment decomposition of $H$.
	Note that $\ell >1$ since otherwise $P_1$ and $P_2$ are two distinct directed paths from $C_0$ to $C_1$, contradicting Lemma~\ref{lem: morelaminar}.
	Suppose that $\ell=2$.
	Since $C_1$ and $C_3$ are common descendants of $C_0$ and $C_2$, there exists a unique greatest common descendant $b$ of $C_0$ and $C_2$ by Lemma~\ref{lem: at most 1 common} and both $C_1$ and $C_3$ are descendants of $b$ (possibly $b=C_i$ for some $i\in [4]$).
	Recall that there is a unique directed path from $C_0$ to each of $C_1$ and $C_3$ by Lemma~\ref{lem: morelaminar}, which must traverse $b$.
	Since $C_0$ is the only node shared by the $(C_0,C_1)$-path $\vec{P_1}$ and the $(C_0,C_3)$-path $\cev{P_4}$, it follows that $C_0=b$.
	Likewise, we can deduce that $C_2=b$, which contradicts that $H$ is a cycle (which do not allow a node repetition).
	
	Suppose that $\ell=3$, and note that $C_{2i}$ is a common ancestor of $C_{2i-1}$ and $C_{2i+1}$ for every $i\in [3]$.
	For each $i\in [3]$, choose an arbitrary clique $C'_i$ which is a sink in $\dH$ and a descendant of $C_{2i-1}$.
	Then, it is easy to see that $C'_i$ is a descendant of $C_{2i-1}$ only for each $i$ by Lemma~\ref{lem: morelaminar}.
	On the other hand, the cliques $C_{2i-1}$ and $C_{2i+1}$ are completely adjacent for every $i\in [3]$, which implies that $C'_1\cup C'_2\cup C'_3$ is a clique because $C'_i\subseteq C_{2i-1}$ for each $i\in [3]$.
	Consider a maximal clique $M$ containing $C'_1\cup C'_2\cup C'_3$ and note that all the nodes of $H$ are descendants of $M$ in $\dH$.
	This contradicts Lemma~\ref{lem: laminar}, which asserts that $\dH[\C_M]$ is an out-tree rooted at $M$, where $\C_M=\{C\in \C(G):C\subseteq M\}$.
	This completes the proof of claim.
\end{proof}

From the laminar structure of $\dH[\C_M]$, we can observe that any pair of nodes are incomparable in $\dH$ if they do not belong to the same segment.

\begin{LEM}\label{lem:incomparable}
	Let $H$ be a cycle of $\dH$ with the shortest segment length with a segment decomposition $$H=C_0,\vec{P_1}, C_1, \cev{P_2},C_2, \cdots ,C_{2\ell-1}, \cev{P_{2\ell}}, C_{2\ell}(=C_0).$$
	Then for any two nodes $C,C'$ of $H$, $C$ and $C'$ are incomparable unless they belong to the same segment of $H$.
\end{LEM}
\begin{proof}
	Suppose not, that is, $C$ and $C'$ are comparable while they do not belong to the same segment of $H$.
	Without loss of generality, we may assume that there exists a directed path $P$ from $C$ to $C'$ and the internal nodes of $P$ avoid $H$.
	We also assume that the segment $P_i$ contains the node $C$ but not $C'$, and the segment $P_j$ contains $C'$ and not $C$, with $j\leq i+\ell$.
	Now, the cycle obtained by traversing the segments $P_i,\ldots , P_j$ and the path $P$ bypass at least $\ell$ segments, and thus its segment length is at most $\ell$.
	This contradicts the choice of $H$.
\end{proof}

\begin{LEM}\label{lem:nocommonanc}
	Let $H$ be a cycle of $\dH$ with the shortest segment length with a segment decomposition $$H=C_0,\vec{P_1}, C_1, \cev{P_2},C_2, \cdots ,C_{2\ell-1}, \cev{P_{2\ell}}, C_{2\ell}(=C_0).$$
	For $i,j\in [\ell]$ with $|i-j|\geq 2$, there is no common ancestor of $C_{2i-1}$ and $C_{2j-1}$ in $\dH$.
\end{LEM}
\begin{proof}
	Suppose not and without loss of generality, there exists $i,j$ with $i+2\leq j\leq i+\ell/2$ such that $C_{2i-1}$ and $C_{2j-1}$ have a common ancestor in $\dT$.
	Let $C$ be the least common ancestor of $C_{2i-1}$ and $C_{2j-1}$ and note that there are internally vertex-disjoint paths, say $P$ and $Q$, from $C$ to $C_{2i-1}$ and $C_{2j-1}$ respectively.
	By Lemma~\ref{lem:incomparable}, $P$ is disjoint from all segments except for the two segments $\vec{P}_{2i-1}$ and $\cev{P}_{2i}$.
	Likewise, $Q$ is disjoint from all segments except for $\vec{P}_{2j-1}$ and $\cev{P}_{2j}$.
	Now observe that there is a cycle $H'$ contained in the segments $\vec{P}_{2i-1},\ldots , \vec{P}_{2j-1},\cev{P}_{2j}$ and the directed paths $P$ and $Q$.
	It is easy to check that the segment length of $H'$ is at most $2(j-i)+2$, which is at most $\ell +2<2\ell$.
	This contradicts the choice of $H$ as a cycle with the shortest segment length.
\end{proof}

\subsection{Constructing  inter-clique digraphs for ($C_4$, \gem)-free graphs}\label{subsec:constructingICD}

Throughout the current subsection, $\dH$ denotes the Hasse diagram $(\C(G),\subseteq)$.
For a maximal clique $M\in \M$, we denote by $\C_M$ the sub-collection of $\C(G)$ of comprising all cliques contained in $M$; that is, $\C_M=\{C\in \C: M\in \rt(C)\}$.

In order to apply the constant-factor approximation algorithm for \FVSP, we need to construct the inter-clique digraph of the input graph $G$, or equivalently the Hasse diagram $\dH$ of $(\C(G),\subseteq)$, in polynomial time.
As an arbitrary graph can have prohibitively many maximal cliques, we cannot expect a polynomial-time algorithm for general graphs.
Instead, we present a polynomial-time algorithm for ($C_4$, \gem)-free graphs.
Such an algorithm is good enough when we aim for a constant-factor approximation algorithm for \PD.

We shall use as the building blocks the partition $\Z$ of the vertex set of $G$ into true twin classes.
Notice that we do not know $\C(G)$ in advance, and actually it is the gist of our algorithm to discover all elements of $\C_M$ while avoiding enumerating all possible subsets of $\M(G)$ containing $M$.
We also want to evade enumerating all possible unions of twin classes to discover a clique $C\in \C(G)$.

Alternatively, we  build $\dH$, given $\M(G)$, $\Z$ and $\rt(\Z)$, in a bottom-to-top manner: that is, we identify a clique $C\in \C(G)$ upon the condition that all its immediate descendants have already been identified.
Since the sinks of $\dH$ are  twin classes by Lemma~\ref{lem:lastnode4twin}, the base case of this approach is valid.
Two key observations pave the way to the polynomial runtime of this bottom-up approach.
First, thanks to Lemma~\ref{lem:twoenough},  any node $C$ of $\dH$ can be `discover' (as an element of $\C(G)$) by considering at most two of its immediate descendants, provided that those immediate descendants have been already discovered.
Secondly, we need a polynomial upper bound on the size of $\C(G)$.
This upper bound is conveniently provided by the laminarity of $\C_M$ for ($C_4$, \gem)-free graphs, see Lemma~\ref{lem: laminar}.

\begin{LEM}\label{lem:C4free}\cite{Farber89a,TsukiyamaIAS77}
	If $G$ does not have an induced $C_4$, there are at most $n^2$ maximal cliques in $G$.
	Moreover, the maximal cliques can be enumerated with polynomial delay.
\end{LEM}

\begin{PROP}\label{thm:ptimeT}
	There is a polynomial-time algorithm which, given a ($C_4$, \gem)-free graph $G$, constructs the Hasse diagram $\dH$ of $(\C(G),\subseteq)$.
\end{PROP}
\begin{proof}
	To begin with, the algorithm creates the collection $\M$ of all maximal cliques of $G$.
	This can be done in polynomial time due to Lemma~\ref{lem:C4free}.
	Next, one obtains the partition $\Z$ of $V$ into true twin classes, which can be clearly done in polynomial time.
	Furthermore, the collection $\rt(\Z)=\{\rt(Z): Z\in \Z\}$ can be efficiently computed by checking the containment relation between the twin classes in $\Z$ and the maximal cliques list $\M(G)$.
	
	Observe that for certain cliques $C\in \C(G)$, $\rt(C)$ is already contained in $\rt(\Z)$.
	
	\begin{CLAIM}\label{claim:wheninZ}
		If $C\in \C(G)$ is a sink or has a unique immediate descendant in $\dH$, then $\rt(C) \in \rt(\Z)$.
	\end{CLAIM}
	\begin{proofofclaim}
		If $C$ is a sink in $\dH$, then $C$ is a twin class itself by Lemma~\ref{lem:lastnode4twin} and thus contained in $\Z$.
		Suppose that $C$ has a sole immediate descendant $C'$ in $\dH$.
		Then $C\setminus C'\neq \emptyset$ and forms a single (true) twin class due to Lemma~\ref{lem:lastnode4twin} because no vertex of $C\setminus C'$ appears in a proper descendant of $C$.
		Again the same lemma and the fact $\rt(C)\subseteq \rt(C\setminus C')$ implies $\rt(C)=\rt(C\setminus C')$.
		From $C\setminus C'\in \Z$, it follows that $\rt(C)=\rt(C\setminus C')\in \rt(\Z)$.
	\end{proofofclaim}
	
	Let $\RR^{\cap,0}:=\rt(\Z)$.
	For $i\geq 1$, we define $\RR^{\cap,i}$ recursively as follows: $$\RR^{\cap,i}:= \RR^{\cap,i-1}\cup \{R\cap R': R, R'\in \RR^{\cap,i-1}\}.$$
	Let the \emph{height of a node} $v$ of an acyclic digraph $\dG$ be the length of a longest directed path from $v$ to a sink in $\dG$.
	The height of $\dG$ is defined as the maximum over the heights of all nodes of $\dG$.
	We claim that there exists $s$ such that $\RR^{\cap,s}$ coincides with $\rt(\C(G)):=\{\rt(C):C\in \C(G)\}$, where $s$ is the height of $\dH$.
	
	\begin{CLAIM}\label{claim:finite}
		$\RR^{\cap,s}=\rt(\C(G))$, where $s$ is the height of $\dH$.
	\end{CLAIM}
	\begin{proofofclaim}
		It suffices to prove the following for each $i\geq 0$: for  any node $C$ at height $i$ in $\dH$, we have $\rt(C)\in \RR^{\cap,i}$.
		By Claim~\ref{claim:wheninZ}, this is true for $i=0$.
		Consider a node $C$ at height $i>0$.
		If $C$ has a single immediate descendant, then $\rt(C)\in \RR^{\cap,0}\subseteq \RR^{\cap,i}$ by Claim~\ref{claim:wheninZ}.
		Suppose that $C$ has (at least) two immediate descendants $C_1,C_2$ in $\dH$.
		By induction hypothesis and because of the fact that the height of the immediate descendants of $C$ is at most $i-1$, we have $\rt(C_1), \rt(C_2)\in \RR^{\cap,i-1}$.
		Therefore, we have $\rt(C_1)\cap \rt(C_2)\in \RR^{\cap,i}$ by definition.
		As $\rt(C_1)\cap \rt(C_2)=\rt(C)$ by Lemma~\ref{lem:twoenough}, it holds that $\rt(C)\in \RR^{\cap,i}$ as claimed.
	\end{proofofclaim}
	
	\begin{CLAIM}\label{claim:height}
		The height of $\dH$ is at most $n$.
	\end{CLAIM}
	\begin{proofofclaim}
		We show that the height of $\dH$ is at most $\abs{\Z}$.
		Indeed it suffices to prove that the height of $\dH[\C_M]$ is at most $\abs{\Z}$ for an arbitrary maximal clique $M$ as any source-to-sink path resides in $\dH[\C_M]$ for some $M$.
		By Lemma~\ref{lem: laminar}, the subdigraph $\dH[\C_M]$ is an out-tree for each maximal clique $M$.
		Let $L$ be the set of all leaf nodes in $\dH[\C_M]$.
		Then the height of $\dH[\C_M]$ is at most $\log \abs{L} + (\abs{\Z}-\abs{L})$ as $\log \abs{L}$ counts the maximum number of branch nodes (i.e., nodes with at least two immediate descendants) and $\abs{\Z}-\abs{L}$ is a trivial upper bound on the number of internal nodes with a single immediate descendant.
		This completes the proof.
	\end{proofofclaim}
	
	As we compute $\RR^{\cap,i+1}$ from $\RR^{\cap,i}$ repeatedly, we need a guarantee that the sizes of the computed sets $\RR^{\cap,i}$ do not grow exponentially.
	The next claim ensures this property thanks to the laminarity of $\C_M$.
	
	\begin{CLAIM}\label{claim:noexponential}
		$\abs{\C_M}\leq 2n$ for each maximal clique $M$ and $\abs{\C(G)}\leq 2n^3$.
	\end{CLAIM}
	\begin{proofofclaim}
		It suffices to prove the $\abs{\C_M}\leq 2\abs{\Z}$ and the second equality follows from Lemma~\ref{lem:C4free} and the trivial bound on $\abs{\Z}$.
		By Lemma~\ref{lem: laminar}, the subdigraph $\dH[\C_M]$ is an out-tree for each maximal clique $M$.
		Let us bound the number of nodes in $\dH[\C_M]$.
		The number of nodes which are leaf nodes or internal nodes with a single immediate descendant is bounded by $\abs{\Z}$ by Claim~\ref{claim:wheninZ}.
		The remaining nodes are internal nodes of with at least two immediate descendants, which is bounded by the number of leaf nodes, and thus by $\abs{\Z}$.
	\end{proofofclaim}
	
	We complete the algorithm description and its runtime analysis.
	Due to Claim~\ref{claim:noexponential}, we can compute each $\RR^{\cap,i}$ in polynomial time and $\rt(\C(G))$ can be computed in polynomial time by Claims~\ref{claim:finite} and~\ref{claim:height}.
	As we compute $\RR^{\cap,i}$, the containment relations amongst the elements of $\RR^{\cap,i}$ can be determined as well.
	Finally, observe that $\dH$ can be obtained from the Hasse diagram of the poset $(\rt(\C(G)),\subseteq)$ by reversing the direction of each arc due to Lemma~\ref{lem:labelrelation}.
\end{proof}

\subsection{Reduction from \PD\ to \FVSP}\label{subsec:final}

Let $G=(V,E)$ be a ($C_4$, \gem)-free graph with vertex weight $\omega^o:V\rightarrow \mathbb{R}_+ \cup \{0\}$.
We want to reduce the instance $(G,\omega^o)$ of  \PD\ to an instance $(\dT,\omega)$ of \FVSP\ so that a solution to the former can be translated to a solution to the latter of the same weight and vice versa.
On the way to define such an instance of \FVSP, we need a few notations.

Let $\dH$ be the Hasse diagram of $(\C(G),\subseteq)$ and let $\dT=(N,A)$ be an inter-clique digraph isomorphic to $\dH$ with an arc-preserving mapping 
\[
\gamma: \C(G)\rightarrow N.
\]
That is, $(C,C')$ is an arc of $\dH$ if and only if $(\gamma(C),\gamma(C'))$ is an arc of $\dT$.
If $\gamma(C)=x$ for some $C\in \C(G)$ and $x\in N$, we may refer to $C$ as the clique corresponding to the node $x$ of $\dT$ instead of invoking the bijection $\gamma$.

Notice that the canonical clique can be construed as a function which maps each vertex $v$ of $G$ to the clique $C\in \C(G)$ such that $\rt(v)=\rt(C)$.
We define a mapping $C^{-1}:\C(G)\rightarrow 2^V$ so that it maps each clique $C$ of $\C(G)$ to its preimage under the canonical clique as a function from $V$ to $\C(G)$: if there is no vertex $v\in V$ with $C(v)=C$, then the preimage of $C$ under the canonical clique is $\emptyset$.
Let the mapping $\phi: V\rightarrow N$ be the composition of $\gamma$ and the canonical clique as a function; that is, for every $v\in V$, we have
\[
\phi(v)=\gamma(C(v)).
\]
Likewise, $\iphi:N\rightarrow 2^V$ is defined as the composition of $C^{-1}$ and $\gamma^{-1}$, namely for every $x\in N$ we let 
\[
\iphi(x)=C^{-1}(\gamma^{-1}(x)).
\]
We remark that $\{\iphi(x):x\in N,\ \iphi(x)\neq\emptyset\}$ is a partition of $V$ by Lemma~\ref{lem: labeling}.
Now the node weight function $\omega:N\rightarrow \mathbb{R}_+\cup \{0\}$ is defined as follows; for every  $x\in N$, $$\omega(x) :=\sum_{v\in \iphi(x)} \omega^o(v).$$
In other words, $\omega(x)$ is the sum of weights of vertices whose canonical clique corresponds to the node $x$ in $\dT$.

For a set of nodes $R$ of $\dT$, the \emph{closure} of $R$, denoted as $R^{*}$, is a minimal superset of $R$ for which the following holds:
\begin{enumerate}[(a)]
	\item all descendants of $R$ of weight zero are contained in $R^*$, 
	\item if all immediate descendants of a node $v$ are contained in $R^*$ and $\iphi(v)= \emptyset$, then $v\in R^*$.
\end{enumerate} 
It is tedious to see that there is a unique closure of a node set $R$, and thus the closure is well-defined.

A node set $R$ is \emph{\closed} in $\dT$ if $v\in R$ implies all descendants of $v$ is in $R$ as well.
We point out that a \closed\ set $R$ is not necessarily a closure of itself (i.e., $R=R^*$) because it may violate the condition (b).
Conversely, a set $R$ which is the closure of itself is not necessarily \closed\ 
as there might a node with non-zero weight which is a descendant of some node of $R$, but not contained in $R$.
Having defined an instance $(\dT,\omega)$ of \FVSP\ from the instance $(G,\omega^o)$ of \PD, the main result of this subsection is presented in the next statement.

\begin{PROP}\label{prop:correctalgo}
	Let $G=(V,E)$ be a ($C_4$, \gem)-free graph with vertex weight $\omega^o:V\rightarrow \mathbb{R}$.
	Let $\dT=(N,A)$ be an inter-clique digraph of $G$ with an arc-preserving mapping $\gamma:\C(G)\rightarrow N$ and with node weight $\omega:\C \rightarrow \mathbb{R}_+\cup\{0\}$, such that $$\omega(x) :=\sum_{v\in \iphi(x)} \omega^o(v),$$ where we define
	\begin{align*}
		\phi(v)&=\gamma(C(v))  &\text{for every } v\in V \\
		\iphi(x)&=C^{-1}(\gamma^{-1}(x))   &\text{for every } x\in N.
	\end{align*}
	Then the following two statements hold.
	\begin{enumerate}[(1)]
		\item For any minimal ptolemaic deletion set $S\subseteq V$, (i) $\phi(S)^{*}$ is \closed\ in $\dT$, (ii) $\und(\dT\setminus \phi(S)^{*})$ is a forest, and (iii) $\sum_{x\in \phi(S)^*} \omega(x)=\sum_{v\in S}\omega^o(v)$.
		\item For any $R\subseteq N$ such that (i) $R$ is downward-closed in $\dT$, and (ii) $\und(\dT\setminus R)$ is a forest, $\iphi(R)$ is a ptolemaic deletion set of $G$ of weight $\sum_{x\in R}\omega(x)$.
	\end{enumerate}
\end{PROP}
\begin{proof}
	We first prove (1)-(i).
	We first observe that if $S$ is a minimal deletion set, $S$ contains the canonical clique $C(v)$ of $v$ whenever $S$ contains $v\in V$.
	
	\begin{CLAIM}\label{claim:allinS}
		If $S\subseteq V$ is a minimal ptolematic deletion set, then $C(v)\subseteq S$ whenever $v\in S$.
		Consequently, $\iphi(x)\subseteq S$ for every $x\in \phi(S)$.
	\end{CLAIM}
	\begin{proofofclaim}
		Suppose $C(v)\not \subseteq S$ for some $v\in S$.
		Since $G$ is ($C_4$, \gem)-free, by (3) of Theorem~\ref{Howorka1981}, $G\setminus S$ is ptolemaic if and only if $G\setminus S$ is chordal.
		Since $S$ is minimal, $G\setminus(S\setminus\{v\})$ has a hole $H$ intersecting $v$.
		By the assumption, there exists $v'\in C(v)\setminus S$.
		However, Lemma~\ref{lem: hereditary} implies that $G[(V(H)\setminus\{v\})\cup\{v'\}]$ contains a hole and thus $G\setminus S$ contains a hole, a contradiction.
		The second statement is immediate from the first statement.
	\end{proofofclaim} 
	
	Consider a vertex $v\in S$ of $G$ and an arbitrary descendant $x$ of $\phi(v)$ in $\dT$.
	We claim that $x\in \phi(S)^{*}$.
	If $\iphi(x)=\emptyset$, then by definition $\omega(x)=\sum_{v\in \emptyset} \omega^o(v)=0$ and thus the claim trivially holds by definition of $\phi(S)^{*}$.
	Otherwise, let $w\in \iphi(x)$ and we have \[\iphi(x)\subseteq \gamma^{-1}(x)\subseteq \gamma^{-1}(\phi(v))=C(v) \subseteq S,\] where the first containment comes from that $\iphi(x)$ is a twin class contained in the clique $\gamma^{-1}$, the second one from the ancestor-descendant relation between $x$ and $\phi(v)$, and the last containment is due to Claim~\ref{claim:allinS}.
	Therefore, $w \in\iphi(x) \subseteq S$ which implies $x\in \phi(S)$.
	This proves that $\phi(S)^{*}$ is downward-closed in $\dT$.
	
	To see that (1)-(ii), let $H$ be a cycle of $\dT\setminus \phi(S)^*$ with the least segment length and let $$x_0,\vec{P_1}, x_1, \cev{P_2},x_2, \cdots ,x_{2\ell-1}, \cev{P_{2\ell}}, x_{2\ell}(=x_0)$$ be a segment decomposition of $H$.
	Consider the cliques $\gamma^{-1}(x_{2i-1})\setminus S$ of $G$ for $i\in [\ell]$.
	We first argue that for every $i\in [\ell]$, there exists a vertex $v_i\in \gamma^{-1}(x_{2i-1})\setminus S$ of $G$.
	Suppose this is not the case, i.e., there exists $i$ such that  $\gamma^{-1}(x_{2i-1})\subseteq S$.
	As we know already that (1)-(i) holds, the fact that $x_{2i-1}\notin \phi(S)^*$ while the clique $\gamma^{-1}(x_{2i-1})$ is contained in $S$ implies $\iphi(x_{2i-1})=\emptyset$.
	It also follows from $\gamma^{-1}(x_{2i-1})\subseteq S$ that for every descendant $y$ of $x_{2i-1}$ in $\dT$ satisfies $\phi^{-1}(y)=\emptyset$ or $y\in \phi(S)$.
	Then, the property (a) of the closure $\phi(S)^*$ imposes $y$ to be included in $\phi(S)^*$, which in turn imposes $x_{2i-1}\in \phi(S)^*$ by the property (b).
	This contradicts the assumption that $H$ is a cycle in $\dT\setminus \phi(S)^*$.
	Therefore, we can choose a vertex $v_i\in \gamma^{-1}(x_{2i-1})\setminus S$ of $G$ for each $i\in [\ell]$.
	
	Next, we observe that all $v_i$'s are distinct.
	Indeed, suppose that $v_i=v_j$ for $i\neq j$, and without loss of generality we may assume that $1\leq i < j\leq \ell$.
	Then the canonical clique $C(v_i)$ is a common descendant of $\gamma^{-1}(x_{2i-1})$ and $\gamma^{-1}(x_{2j-1})$, or equivalently, $\phi(v_i)$ is a common descendant of $x_{2i-1}$ and $x_{2j-1}$.
	Let $x^*$ be the greatest common descendant of $x_{2i-1}$ and $x_{2j-1}$ in $\dT$, which is unique by Lemma~\ref{lem: at most 1 common}.
	Let $P$ and $Q$ be the directed $(x_{2i-1},x^*)$-path and the directed $(x_{2j-1},x^*)$-path.
	Due to Lemma~\ref{lem:incomparable}, both directed paths are disjoint from $H$ except from the two starting vertex $x_{2i-1}$ and $x_{2j-1}$.
	Therefore, we can obtain a new cycle $H'$ from $H$ by replacing the subpath of $H$ consisting of segments $x_{2i-1},\cev{P}_{2i},\ldots , \vec{P}_{2j-1},x_{2j-1}$ by $x_{2i-1},P,x^*,Q,x_{2j-1}$.
	Note that now the concatenation of $\vec{P}_{2i-1}$ and $P$ yields a directed path, and likewise, the concatenation of $\cev{P}_{2j}$ and $Q$ yields a directed path.
	Therefore, the segment length of $H'$ is shorter than that of $H$ by at least two.
	This contradicts the choice of $H$.
 	
	Furthermore, $v_i$ and $v_{i+1}$ are adjacent because the cliques $\gamma^{-1}(x_{2i-1})$ and $\gamma^{-1}(x_{2(i+1)-1})$ are complete to each other in $G$ due to the existence of common ancestor $x_{2i}$ in $\dT$.
	That is, $J=v_1,\ldots , v_{\ell}, v_1$ forms a cycle, and its length is at least four by Lemma~\ref{lem:four}.
	Furthermore, Lemma~\ref{lem:nocommonanc} implies that $J$ is a hole, which altogether avoids $S$ because of our choice of $v_i$ as a vertex of $\gamma^{-1}(x_{2i-1})\setminus S$.
	This contradicts the assumption that $S$ is a ptolemaic deletion set, which proves (1)-(ii).

	Lastly, (1)-(iii) follows from 
	\[
	\sum_{x\in \phi(S)^*}\omega(x) = \sum_{x\in \phi(S)}\omega(x)=\sum_{x\in \phi(S)} \sum_{v\in \iphi(x)}\omega^o(v) = \sum_{v\in S}\omega^o(v),
	\]
	where the first equality is from the definition of closure, the second from the definition of the node weight $\omega$, and the last equality is because $S$ is partitioned into $\{\{\iphi(x):x\in \phi(S)\}$ by Claim~\ref{claim:allinS}.

	To see (2), suppose that for a node set $R$ of $\dT$ (i) $R$ is downward-closed in $\dT$, and (ii) $\und(\dT\setminus R)$ is a forest while $\bigcup_{x\in R} \iphi(x)$ is not a ptolemaic deletion set of $G$.
	Let $H=v_1,\ldots , v_s, v_1$ be a hole of length $s\geq 5$ in $G\setminus \bigcup_{x\in R} \iphi(x)$.
	Consider the canonical cliques $C(v_1),\ldots,C(v_s)$ and their corresponding nodes $x_1,\ldots , x_s$ in $\dT$.
	The adjacency of $v_i$ and $v_{i+1}$ ensures that $x_i$ and $x_{i+1}$ has a common ancestor for all $i\in [s]$, where $s+1=1$.
	Furthermore, none of the nodes from these common ancestors is contained in $R$ since otherwise, some $x_i$ must belong to the \closed\ set $R$.
	This, however, means that $x_1,\ldots , x_s$ are contained in a closed walk of $\dT \setminus R$, contradicting (ii).
	We conclude that $\bigcup_{x\in R} \iphi(x)$ is a ptolemaic deletion set of $G$.
	Finally, the weight of the ptolemaic deletion set is 
	\[
	\sum_{v\in \bigcup_{x\in R} \iphi(x)}\omega^o(v)= \sum_{x\in R}\sum_{v\in \iphi(x)}\omega^o(v)=\sum_{x\in R}\omega(x),
	\]
	because the first equality holds as $\iphi(x)\cap \iphi(y)=\emptyset$ whenever $x\neq y$, and the second equality holds by definition of the node weight function $\omega$.
\end{proof}

\begin{THM}
	There is a polynomial-time algorithm which, given a graph $G=(V,E)$ with vertex-weight $\omega^o: V\rightarrow \mathbb{R}_+\cup \{0\}$, returns a ptolemaic deletion set $S\subseteq V$ of weight at most $68\cdot \opt_{pto}$, where $\opt_{pto}$ is the minimum weight of a ptolematic deletion set of $G$.
\end{THM}
\begin{proof}
	We skip the trivial runtime analysis.
	For simplicity, we write $w^o(v)$ as $w^o_v$.
	In order to turn the input graph into a ($C_4$, \gem)-free graph, we employ a rounding algorithm using an optimal fractional solution to the next linear programming ($\lp$) relaxation.

	\begin{align}
	\mbox{Minimize }&  \sum_{v \in V} \omega^o_v x_v \nonumber \\
	\mbox{Subject to }
	& \sum_{v\in A} x_v \geq 1 && \forall A\subseteq V \text{ such that $G[A]$ is isomorphic to $gem$ or $C_4$} \label{eq:C4gem0} \\
	& 0 \leq x \leq 1. \nonumber 
	\end{align}
	
	Let $x^*$ be an optimal solution to the above $\lp$ and let $X\subseteq V$ be the vertex set consisting of all $v$'s with $x^*_v\geq 0.2$.
	Note that the weight of $X$ is
	\[ \omega^o(X)=\sum_{v\in X}\omega^o_v\cdot 1 \leq \sum_{v\in X}\omega^o_v \cdot 5x^*_v \leq 5 \opt_{pto},\]
	where the second and the third inequalities holds due to the construction of $X$ and that an integral solution of weight $\opt_{pto}$ is feasible to the above $\lp$.
	
	Now, we consider the graph $G'$ obtained by removing the vertices of $X$ from $G$ and notice that $G'$ is ($C_4$, \gem)-free.
	Each vertex of $G'$ inherits its weight $\omega^o_v$ in $G$.
	We construct an inter-clique digraph $\dT=(N,A)$ of $G'$ with a node-weight $\omega$ as in Proposition~\ref{prop:correctalgo}; notice that the inter-clique digraph $\dT$ can be constructed in polynomial time by the algorithm of Proposition~\ref{thm:ptimeT}.
	The node set $\anc(x)$ forms an in-tree rooted at $x$ due to Lemma~\ref{lem: morelaminar}, which means that $(\dT,\omega)$ is a legitimate instance to \FVSP.
	Therefore we can apply the algorithm of Theorem~\ref{thm:approxFVSP} and attain a solution $R\in N$ such that $R$ is downward-closed in $\dT$, $\und(\dT \setminus R)$ is a forest, and $\omega(R)\leq 63\opt_{fvsp}$.
	Here $\opt_{fvsp}$ is the minimum weight of a solution to \FVSP.
	
	We claim that $\bigcup_{x\in R} \iphi(x)\cup X$ is a ptolemaic deletion set of $G$ with weight at most $68\opt_{pto}$.
	Indeed, $\bigcup_{x\in R} \iphi(x)$ is a ptolemaic deletion set of $G'$ with weight $\sum_{x\in R}\omega(x)$ by (2) of Proposition~\ref{prop:correctalgo} and thus its weight is at most  $63\opt_{fvsp}$.
	Finally, from (1) of Proposiotion~\ref{prop:correctalgo} we know that $\opt_{fvsp}\leq \opt_{pto}$.
	This proves that claim, thus the main statement.
\end{proof}

\section{Constant-factor approximation algorithm}\label{sec: constant-factor approx}

In this section, we consider \FVSP introduced in Section~\ref{subsec:tech_fvsp}.

\vskip 0.2cm
\noindent
\fbox{\parbox{0.97\textwidth}{
\FVSP \\
\textbf{Input :} An acyclic directed graph $\dG = (V, A)$, where each vertex $v$ has weight $\omega_v \in \R^+ \cup \{ 0 \}$. 
For each $v \in V$, the subgraph induced by $\anc(v)$ is an in-tree rooted at $v$. 
\\
\textbf{Question :} Delete a minimum-weight vertex set $S \subseteq V$ such that (1) $v \in S$ implies $\des(v) \subseteq S$, (2) $\und(\dG \setminus S)$ is a forest. 
}}
\vskip 0.2cm

It is a variant of \textsc{Undirected Feedback Vertex Set} on $\und(\dG)$, with the additional precedence constraint on $S$ is captured by the direction of arcs in $A$.
The main result of this section is an $O(1)$-approximation algorithm for this problem.

\begin{THM}\label{thm:approxFVSP}\label{thm:approx}
	There is a polynomial-time $63$-approximation algorithm for \FVSP.
\end{THM}

We consider the following linear programming ($\lp$) relaxation.
The relaxation variables are $\{ z_v \}_{v \in V}$, where $z_v$ is supposed to indicate whether $v$ is deleted or not, as well as $\{ x_{ue} \}_{e \in A, u \in e}$, where $x_{ue}$ is supposed to indicate that in the resulting forest $\und(\dG \setminus S)$ rooted at arbitrary vertices, whether $e$ is the edge connecting $u$ and its parent.

\begin{align}
	\mbox{Minimize }&  \sum_{v \in V} \omega_v z_v \nonumber \\
	\mbox{Subject to }
	& z_v + x_{ue} + x_{ve} = 1 && \forall e = (u, v) \in A \label{eq:lp-1} \\
	& z_v + \sum_{e \ni v} x_{ve} \leq 1 && \forall v \in V\label{eq:lp-2}  \\
	& z_u \leq z_v && \forall e = (u, v) \in A \nonumber \\
	& 0 \leq x, z \leq 1. \nonumber 
\end{align}

Let $\opt$ be the weight of the optimal solution, and $\lp \leq \opt$ be the optimal value of the above LP.
After solving the LP, we perform the following rounding algorithm.
It is parameterized by three parameters $\eps, \alpha, \beta \in (0, 1)$ that satisfy
\begin{align}
	2\alpha &\geq 1 + \eps, \label{eq:param1} \\
	3(1-\beta) &\geq 1 + 8\eps. \label{eq:param2}
\end{align}
(The final choice will be $\eps \approx 0.029, \alpha \approx 0.514, \beta \approx 0.588$.)
For notational convenience, let $\bx_{ue} := 1 - x_{ue}$.
Also, for each $e = (u, v) \in A$, let $y_{e} = z_v - z_u$.
Each vertex $v \in V$ maintains a set $L_v \subseteq A$.
Initially, all $L_v$'s are empty.
\begin{enumerate}[(i)]
	\item Delete all vertex $v$ with $z_v \geq \eps$.
	\item Sample $\theta$ uniformly at random from the interval $[\alpha, \beta]$.
	\item For each $e = (u, v) \in A$, 
	\begin{itemize}
		\item If $\theta \in [\bx_{ve} - y_{e}, \bx_{ve}]$, delete $\des(v)$. Say $v$ is {\em directly deleted by $e$.}
		\item Otherwise, 
		\begin{itemize}
			\item If $\theta > \bx_{ve}$, then add $e$ to $L_v$ and say $v$ {\em points to} $e$.
			\item If $\theta > \bx_{ue}$, then add $e$ to $L_u$ and say $u$ {\em points to} $e$.
		\end{itemize}
	\end{itemize}
\end{enumerate}

Though the above rounding algorithm is stated as a randomized algorithm, it is easy to make it deterministic, because there are at most $O(m)$ subintervals of $[\alpha, \beta]$ such that two $\theta$ values from the same interval behave exactly the same in the rounding algorithm.

We first analyze the total weight of deleted vertices.
In Step (i), we delete all vertices whose LP value $z_v \geq \eps$, so the total weight of deleted vertices in Step (i) is at most $\lp / \eps$.
The following lemma bounds the weight of vertices deleted in Step (iii) by at most $2\lp / (\beta - \alpha)$.

\begin{LEM}
	For each $v \in V$, $\Pr[v\mbox{ is deleted in Step (iii)}] \leq \frac{2z_v}{\beta - \alpha}$.
\end{LEM}
\begin{proof}
	Due to Step (i), we can assume that every vertex $v$ satisfies $z_v < \eps$ and each arc $e$ satisfies $y_e < \eps$.
	
	Fix a vertex $v \in V$.
	Let $\dT = (V(\dT), A(\dT))$ be the subgraph of $\dG$ induced $\anc(v)$.
	By the definition of \FVSP, $\dT$ is an in-tree rooted at $v$.
	We first prove the following claim that if we consider any directed path $(u_0, \dots, u_{k})$ of $\dT$ and the value of $x_{u_i, (u_{i-1}, u_i)}$ that $u_i$ gives to its incoming edge $(u_{i-1},u_i)$, the value at the end $(i = k)$ is almost as large as the value at the beginning $(i = 1)$.
	
	\begin{CLAIM}\label{claim:chain}
		Let $(u_0, \dots, u_k)$ be a directed path in $\dT$ and $e_i = (u_{i- 1}, u_i)$.
		Then for any $i \in [k]$, $x_{u_i e_i} \geq x_{u_1e_1} - (z_{u_i} - z_{u_1}) \geq x_{u_1 e_1} - \eps$.
	\end{CLAIM}
	\begin{proofofclaim}
		The proof proceeds by induction.
		The base case $i = 1$ is obviously true.
		When the claim holds for $i - 1$, the constraint~\eqref{eq:lp-2} of the LP (for $u_{i-1}$) implies 
		\[
		x_{u_{i-1} e_{i-1}} + z_{u_{i-1}} + x_{u_{i-1} e_{i}} \leq 1, 
		\]
		and the constraint~\eqref{eq:lp-1} of the LP implies (for $e_i$)
		\[
		z_{u_{i}} + x_{u_{i-1} e_{i}} +  x_{u_{i} e_{i}} = 1.
		\]
		Subtracting the first inequality from the second equality yields
		\[
		x_{u_{i} e_{i}} \geq x_{u_{i-1} e_{i-1}} - (z_{u_{i}} - z_{u_{i - 1}}),
		\]
		which, by the induction hypothesis, is at least 
		\[
		x_{u_{1} e_{1}} - (z_{u_{i-1}} - z_{u_1}) - (z_{u_{i}} - z_{u_{i - 1}}) = 
		x_{u_{1} e_{1}} - (z_{u_i} - z_{u_1}).
		\]
	\end{proofofclaim}
	
	For $e = (w, u) \in A(\dT)$, call $e$ a {\em target} if $\Pr[u\mbox{ is directly deleted by }e] > 0$, which implies $\bx_{ue} - y_{e} < \beta \Rightarrow x_{ue} > 1  - \beta - y_e > 1  - \beta - \eps$.
	For two arcs $e, f \in A(\dT)$, say they are {\em incomparable} if there is no directed path from the tail of one arc to tail of the other in $\dT$ (though they may share the head.)

	\begin{CLAIM}\label{claim:targets}
		There are no three pairwise incomparable targets.
	\end{CLAIM}
	\begin{proofofclaim}
		Assume towards contradiction that there exist three pairwise incomparable targets $e_1 = (w_1, u_1), e_2 = (w_2, u_2), e_3 = (w_3, u_3)$.
		It implies that $x_{u_i e_i} > 1 - \beta - \eps$ for each $i$.
		By Claim~\ref{claim:chain}, for any $i$ and any arc $e' = (w', u') \in A(\dT)$ that has a directed path from $e_i$, we have 
		\begin{equation}
			x_{u'e'} > x_{u_i e_i} - \eps > 1  - \beta - 2\eps.
		\label{eq:incomp}
		\end{equation}
		For each $i \in [3]$, consider the path $P_i$ from $w_i$ to $v$, and let $g_i$ be the last arc of $P_i$ that does not appear in any other $P_j$'s.
		We consider the following two cases depending on how they intersect, and show both cannot happen.

\begin{figure}[h]
\centering
\includegraphics[width=8cm]{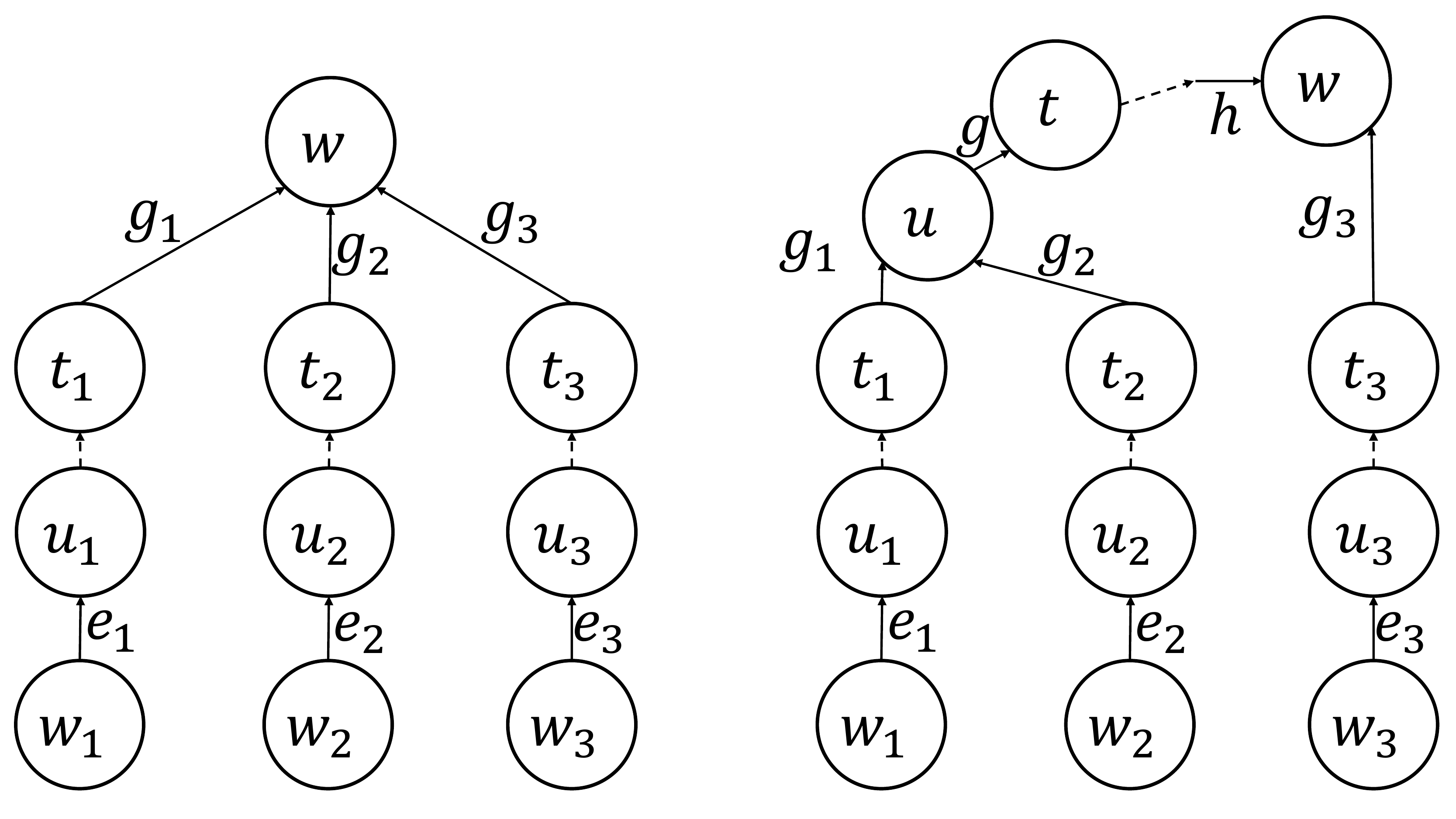}
\caption{Two cases for $g_1, g_2, g_3$.
The left figure shows the case when they all meet at the same vertex $w$.
The right figure shows when $g_1$ and $g_2$ meet first at $u$ and meet $g_3$ with $w$ later.
Real lines indicate an individual arc and dotted lines indicate a directed path.
 } 
\end{figure}

		First, suppose all $g_1, g_2, g_3$ meet at the same vertex $w$; in other words, $g_i = (t_i, w)$ for some $t_i$'s.
		Then, by~\eqref{eq:incomp}, $x_{w g_i} > 1  - \beta - 2\eps$ for each $i$.
		With~\eqref{eq:param2}, it implies $\sum_i x_{w g_i} > 3(1  - \beta) - 6\eps \geq 1$, which violates the constraint~\eqref{eq:lp-2} of the LP.
	
		Finally, without loss of generality, suppose $g_1$ and $g_2$ meet at $u$, which is not incident on $g_3$; in other words, $g_1 = (t_1, u), g_2 = (t_2, u), g_3 = (t_3, w)$ for some $t_i$'s, where $w$ is an ancestor of $u$ in $\dT$ and is the first vertex where all $P_1, P_2, P_3$ intersect.
		Let $\dT$ be the parent of $u$ in the tree $\dT$ ($\dT$ may be equal to $w$), and $g = (u, t)$.
		Then,~\eqref{eq:incomp}, implies $x_{u g_i} > 1  - \beta - 2\eps$ for $i \in \{ 1, 2 \}$, which, combined with the LP constraint~\eqref{eq:lp-2} for $u$, yields
		\[
		x_{ug} < 1 - 2(1  - \beta - 2\eps) = 2\beta - 1 + 4\eps.
		\]
		Together again with the LP constraint~\eqref{eq:lp-1} for $g$, we have
		\[
		x_{tg} > 1 - x_{ug} - z_t > 2 - 2\beta - 5\eps.
		\]
		Let $h$ be the last arc of the path from $u$ to $w$.
		Using Claim~\ref{claim:chain} again, we conclude that $x_{wh} > 2 - 2\beta - 6\eps$.
		Combined with $x_{wg_3} > 1 - \beta - 2\eps$ and $h$ and $g_3$ are different, it implies $x_{wh} + x_{wg_3} > 3 - 3\beta - 8\eps \geq 1$ by~\eqref{eq:param2}, which contradicts the constraint~\eqref{eq:lp-2} of the LP for $w$.
	\end{proofofclaim}
	
	Now we compute the probability that $v$ is deleted by Step (iii) of the rounding algorithm.
	It happens whether $v$ itself is directly deleted or some vertex $u \in \anc(v) = V(\dT)$ is directly deleted by a target $e = (w, u)$.
	By Claim~\ref{claim:targets}, no three targets are pairwise comparable, and by Dilworth's Theorem, all targets are contained in two directed paths $P_1, P_2$ in $\dT$.
	By the choice of the rounding algorithm, for one path $P_1 = (u_0, \dots, u_k = v)$, for each $i \in [k]$, 
	\[
	\Pr[u_i\mbox{ is directly deleted by }(u_{i-1}, u_i)] \leq \frac{ y_{(u_{i-1},u_i)}}{\beta - \alpha} = \frac{ z_{u_i} - z_{u_{i - 1}}}{\beta - \alpha}.
	\]
	Summing over all $i$'s yields 
	\[
	\sum_{i=1}^k \frac{z_{u_i} - z_{u_{i - 1}}}{\beta - \alpha} = \frac{z_{u_k} - z_{u_0}}{\beta - \alpha} \leq \frac{z_v}{\beta - \alpha}.
	\]
	We can apply the same analysis to $P_2$ and use the union bound.
\end{proof}

We now examine structure of the remaining graph after the rounding procedure.
We first show that in the original graph, each arc, if not deleted, is pointed to by at least one of its endpoints.

\begin{CLAIM}\label{claim:point}
	For each $e = (u, v) \in A$, if neither $u$ nor $v$ was deleted during the rounding, $e$ is pointed to by at least one of them.
\end{CLAIM}
\begin{proofofclaim}
	Since $v$ is not deleted, it means $z_v < \eps$, which, by~\eqref{eq:lp-1}, implies that $x_{ue} + x_{ve} > 1 - \eps \Leftrightarrow \bx_{ue} + \bx_{ve} < 1 + \eps$.
	Since $\theta \geq \alpha$, by~\eqref{eq:param1}, either $\theta \geq \bx_{ue}$ or $\theta \geq \bx_{ue}$.
\end{proofofclaim}

The following lemma shows that after the rounding, each connected component (in the undirected sense) has at most one cycle.

\begin{LEM}
	Let $S$ be the set of vertices deleted during the rounding algorithm.
	In each connected component of $\und(\dG \setminus S)$, there is at most one (undirected) cycle.
\end{LEM}
\begin{proof}
	The proof proceeds by examining how vertices can possibly point to adjacent arcs.
	First, the following claim shows that one vertex cannot point to more than two arcs.
	\begin{CLAIM}
		Every vertex $v \in V$ points to at most two arcs.
	\end{CLAIM}
	\begin{proofofclaim}
		Assume towards contradiction that $v$ points to three arcs $e, f,$ and $g$.
		It implies $\bx_{ve}, \bx_{vf}, \bx_{vg}$ are all strictly less than $\theta \leq \beta$, which implies that $x_{ve} + x_{vf} + x_{vg} > 3(1 - \beta)$.
		Since $3(1 - \beta) \geq 1$ by~\eqref{eq:param2}, it contradicts the constraint~\eqref{eq:lp-2} of the LP relaxation.
	\end{proofofclaim}
	Moreover, the following claim constrains the way arcs in a cycle are pointed to by its vertices.
	\begin{CLAIM}
		For every arc $e \in A$, if it is pointed to by exactly one of its endpoint, say $v$, then it is the only arc that $v$ points to.
	\end{CLAIM}
	\begin{proofofclaim}
		We first show $\theta < x_{ve} + z_v$.
		If $e = (u, v)$, the assumption that $u$ does not point to $e$ implies 
		\[
		\theta < \bx_{ue} = 1 - x_{ue} = 1 - (1 - x_{ve} - z_v) = x_{ve} + z_v,
		\]
		where the second equality follows from~\eqref{eq:lp-1}.
		Even when $e = (v, u)$, the assumption that $e$ is not deleted and $u$ does not point to $e$ implies 
		\[
		\theta < \bx_{ue} - y_{e} = \bx_{ue} - (z_u - z_v) = 1 - x_{ue} - z_u + z_v = x_{ve} + z_v,
		\]
		where the last equality follows from the constraint~\eqref{eq:lp-1} of the LP relaxation.
		Therefore, $\theta < x_{ve} + z_v$ in any case.
		
		If $v$ points to any other arc $f$, it implies 
		\[
		\theta > \bx_{vf} = 1 - x_{vf} \geq z_v + x_{ve},
		\]
		where the inequality follows from the constraint~\eqref{eq:lp-2} of the LP relaxation.
		This leads to contradiction, proving the claim.
	\end{proofofclaim}
	
	Therefore, after the rounding, in the remaining graph $\dG \setminus S$, (i) every remaining arc is pointed by at least one of its endpoints, (ii) each vertex points to at most two arcs, and (iii) if one vertex does not point to an arc incident on it, the other endpoint uniquely points to the arc.
	
	Consider an undirected cycle $(v_1, \dots, v_k, v_{k+1})$ in $\und(\dG \setminus S)$ with $v_{1} = v_{k+1}$, so that either $(v_i, v_{i+1})$ or $(v_{i+1}, v_i)$ is in $A$ for every $i \in [k]$.
	Let $\{ v_i, v_{i+1} \}$ denotes an undirected edge.
	If an edge in this cycle is pointed to by only one of its endpoints (without loss of generality, say $\{ v_k, v_1 \}$ is only pointed to by $v_1$), then $v_1$ cannot point to any other edge, so $\{ v_2, v_1 \}$ is uniquely pointed to by $v_2$ by (iii), and this inductively leads to every $\{ v_i, v_{i+1} \}$ uniquely pointed to by $v_{i+1}$ for $1 \leq i < k$.
	Note that all $v_1, \dots, v_k$ cannot point to any edge outside the cycle.
	Even when all edges are pointed to by both endpoints, by (ii), all $v_1, \dots, v_k$ cannot point to any edge outside the cycle.
	
	Assume towards contradiction that there are two undirected cycles $C_1$ and $C_2$ (not necessarily vertex or edge disjoint) in the same connected component of $\und(\dG \setminus S)$.
	If $V(C_1) \cap V(C_2) \neq \emptyset$, there must be a vertex $v \in C_2$ that points to an edge in $C_1 \setminus C_2$.
	This contradicts the above paragraph.
	If $C_1$ and $C_2$ are vertex disjoint, let $(v_1, \dots, v_k)$ be an undirected path from $C_1$ and $C_2$ where $v_1 \in C_1$ and $v_k \in C_2$.
	By the above paragraph, $\{ v_1, v_2 \}$ is uniquely pointed to by $v_2$ and inductively $\{ v_i, v_{i+1} \}$ is uniquely pointed to by $v_{i+1}$.
	But applying the same argument from $\{ v_{k-1}, v_k \}$, $\{v_i, v_{i+1} \}$ must be uniquely pointed to by $v_i$, leading to contradiction.
	Therefore, there must be only one undirected cycle in each connected component.
\end{proof}

After the rounding, each connected component has at most one cycle, so we can easily compute the optimal solution efficiently.
Therefore, we compute a feasible solution that respects the constraints of the \FVSP.
Since the total weights of deleted vertices in each step is at most $\lp / \eps$ in Step (i), at most $2\lp / (\beta - \alpha)$ in Step (iii), and at most $\opt$ in the final cleanup step, the final approximation ratio is 
\[
\frac{1}{\eps} + \frac{2}{\beta - \alpha} + 1 \leq 62.2
\]
by our choice of $\eps = 0.0293258, \alpha = 0.514663, \beta = 0.588465$.

\section{Conclusions}\label{sec: conclusions}

In this paper, we show that \PD admits a polynomial-time constant-factor approximation algorithm.
To this end, we introduce \FVSP, which is a variant of \WFVS with additional constraints including that any solution set must be \closed\, i.e., if a node being in 
a solution propels all its descendants in the input (acyclic) digraph to be in the solution.
To attain an approximation algorithm for \PD from that for \FVSP, we investigate the structure of inter-clique digraphs of graphs which give us a tree-like structure when an input graph is ptolemaic and have a laminar structure when an input graph is ($C_4$, \gem)-free.
\FVSP can be utilized for wider purposes in other parameterized problems, because through our approximation algorithm, one can find a hereditary feedback vertex set of an input graph, where the hereditary property captures the essence of the parameterized problems.
 
As the purpose of this paper, for various graph classes $\mathcal{F}$, it would be interesting to investigate whether \textsc{Weighted $\mathcal{F}$-Deletion} admits a constant-factor approximation algorithm; for instance, \textsc{Chordal Vertex Deletion} or \textsc{$\ell$-Leaf Power Vertex Deletion} for $\ell\geq4$.
For a positive integer $\ell$, a graph $G$ is an $\ell$-leaf power if there is a tree $T$ such that $V(G)$ is equal to the set of leaves of $T$ and $v,w\in V(G)$ are adjacent if and only if $d(v,w)\leq\ell$ in $T$.
We remark that \textsc{$3$-Leaf Power Vertex Deletion} admits a constant-factor approximation algorithm by a reduction to \WFVS, as Dom, Guo, H\"{u}ffner, and Neidermeier~\cite{DomGHN06} did to show the fixed-parameter tractability of \textsc{$3$-Leaf Power Vertex Deletion}.

\providecommand{\bysame}{\leavevmode\hbox to3em{\hrulefill}\thinspace}
\providecommand{\MR}{\relax\ifhmode\unskip\space\fi MR }
\providecommand{\MRhref}[2]{%
  \href{http://www.ams.org/mathscinet-getitem?mr=#1}{#2}
}
\providecommand{\href}[2]{#2}


\begin{thebibliography}{10}

\bibitem{agrawal2016faster}
Akanksha Agrawal, Sudeshna Kolay, Daniel Lokshtanov, and Saket Saurabh, \emph{A
  faster {FPT} algorithm and a smaller kernel for block graph vertex deletion},
  LATIN 2016: Theoretical Informatics, Springer, 2016, pp.~1--13.

\bibitem{agrawal2018feedback}
Akanksha Agrawal, Daniel Lokshtanov, Pranabendu Misra, Saket Saurabh, and
  Meirav Zehavi, \emph{Feedback vertex set inspired kernel for chordal vertex
  deletion}, ACM Transactions on Algorithms (TALG) \textbf{15} (2018), no.~1,
  1--28.

\bibitem{agrawal2018polylogarithmic}
\bysame, \emph{Polylogarithmic approximation algorithms for weighted-f-deletion
  problems}, Approximation, Randomization, and Combinatorial Optimization.
  Algorithms and Techniques (APPROX/RANDOM 2018), Schloss
  Dagstuhl-Leibniz-Zentrum fuer Informatik, 2018.

\bibitem{agrawal2019interval}
Akanksha Agrawal, Pranabendu Misra, Saket Saurabh, and Meirav Zehavi,
  \emph{Interval vertex deletion admits a polynomial kernel}, Proceedings of
  the Thirtieth Annual ACM-SIAM Symposium on Discrete Algorithms, SIAM, 2019,
  pp.~1711--1730.

\bibitem{AEKO2019}
Jungho Ahn, Eduard Eiben, O-joung Kwon, and Sang-il Oum, \emph{A polynomial
  kernel for $3$-leaf power deletion}, arXiv:1911.04249, manuscript, 2019.

\bibitem{BL2006}
Andreas Brandst\"{a}dt and Van~Bang Le, \emph{Structure and linear time
  recognition of 3-leaf powers}, Inform. Process. Lett. \textbf{98} (2006),
  no.~4, 133--138. \MR{2211095}

\bibitem{cao2016linear}
Yixin Cao, \emph{Linear recognition of almost interval graphs}, Proceedings of
  the twenty-seventh annual ACM-SIAM symposium on Discrete algorithms, SIAM,
  2016, pp.~1096--1115.

\bibitem{cao2015interval}
Yixin Cao and D{\'a}niel Marx, \emph{Interval deletion is fixed-parameter
  tractable}, ACM Transactions on Algorithms (TALG) \textbf{11} (2015), no.~3,
  1--35.

\bibitem{CaoM16}
Yixin Cao and D{\'{a}}niel Marx, \emph{Chordal editing is fixed-parameter
  tractable}, Algorithmica \textbf{75} (2016), no.~1, 118--137.

\bibitem{chekuri2016constant}
Chandra Chekuri and Vivek Madan, \emph{Constant factor approximation for subset
  feedback set problems via a new lp relaxation}, Proceedings of the
  twenty-seventh annual ACM-SIAM symposium on Discrete algorithms, SIAM, 2016,
  pp.~808--820.

\bibitem{DomGHN06}
Michael Dom, Jiong Guo, Falk H{\"{u}}ffner, and Rolf Niedermeier, \emph{Error
  compensation in leaf power problems}, Algorithmica \textbf{44} (2006), no.~4,
  363--381.

\bibitem{eiben2018single}
Eduard Eiben, Robert Ganian, and O-joung Kwon, \emph{A single-exponential
  fixed-parameter algorithm for distance-hereditary vertex deletion}, Journal
  of Computer and System Sciences \textbf{97} (2018), 121--146.

\bibitem{Farber89a}
Martin Farber, \emph{On diameters and radii of bridged graphs}, Discrete
  Mathematics \textbf{73} (1989), no.~3, 249--260.

\bibitem{fiorini2010hitting}
Samuel Fiorini, Gwena{\"e}l Joret, and Ugo Pietropaoli, \emph{Hitting diamonds
  and growing cacti}, International Conference on Integer Programming and
  Combinatorial Optimization, Springer, 2010, pp.~191--204.

\bibitem{fomin2012planar}
Fedor~V Fomin, Daniel Lokshtanov, Neeldhara Misra, and Saket Saurabh,
  \emph{Planar f-deletion: Approximation, kernelization and optimal {FPT}
  algorithms}, 2012 IEEE 53rd Annual Symposium on Foundations of Computer
  Science, IEEE, 2012, pp.~470--479.

\bibitem{gupta2019losing}
Anupam Gupta, Euiwoong Lee, Jason Li, Pasin Manurangsi, and Micha{\l}
  W{\l}odarczyk, \emph{Losing treewidth by separating subsets}, Proceedings of
  the Thirtieth Annual ACM-SIAM Symposium on Discrete Algorithms, SIAM, 2019,
  pp.~1731--1749.

\bibitem{heggernes2011parameterized}
Pinar Heggernes, Pim Van't~Hof, Bart~MP Jansen, Stefan Kratsch, and Yngve
  Villanger, \emph{Parameterized complexity of vertex deletion into perfect
  graph classes}, International Symposium on Fundamentals of Computation
  Theory, Springer, 2011, pp.~240--251.

\bibitem{Howorka1981}
Edward Howorka, \emph{A characterization of {P}tolemaic graphs}, J. Graph
  Theory \textbf{5} (1981), no.~3, 323--331. \MR{625074}

\bibitem{jansen2017approximation}
Bart~MP Jansen and Marcin Pilipczuk, \emph{Approximation and kernelization for
  chordal vertex deletion}, Proceedings of the Twenty-Eighth Annual ACM-SIAM
  Symposium on Discrete Algorithms, SIAM, 2017, pp.~1399--1418.

\bibitem{KaplanST94}
Haim Kaplan, Ron Shamir, and Robert~Endre Tarjan, \emph{Tractability of
  parameterized completion problems on chordal and interval graphs: Minimum
  fill-in and physical mapping}, 35th Annual Symposium on Foundations of
  Computer Science, Santa Fe, New Mexico, USA, 20-22 November 1994, {IEEE}
  Computer Society, 1994, pp.~780--791.

\bibitem{kim2017polynomial}
Eun~Jung Kim and O-joung Kwon, \emph{A polynomial kernel for
  distance-hereditary vertex deletion}, Workshop on Algorithms and Data
  Structures, Springer, 2017, pp.~509--520.

\bibitem{KK2018}
\bysame, \emph{Erd{\H{o}}s-{P}{\'{o}}sa property of chordless cycles and its
  applications}, Proceedings of the {T}wenty-{N}inth {A}nnual {ACM}-{SIAM}
  {S}ymposium on {D}iscrete {A}lgorithms, SIAM, Philadelphia, PA, 2018,
  pp.~1665--1684. \MR{3775897}

\bibitem{kim2015linear}
Eun~Jung Kim, Alexander Langer, Christophe Paul, Felix Reidl, Peter Rossmanith,
  Ignasi Sau, and Somnath Sikdar, \emph{Linear kernels and single-exponential
  algorithms via protrusion decompositions}, ACM Transactions on Algorithms
  (TALG) \textbf{12} (2015), no.~2, 1--41.

\bibitem{kratsch2014compression}
Stefan Kratsch and Magnus Wahlstr{\"o}m, \emph{Compression via matroids: a
  randomized polynomial kernel for odd cycle transversal}, ACM Transactions on
  Algorithms (TALG) \textbf{10} (2014), no.~4, 1--15.

\bibitem{LMPPS2020}
Daniel Lokshtanov, Pranabendu Misra, Fahad Panolan, Geevarghese Philip, and
  Saket Saurabh, \emph{A (2 + $\epsilon$)-factor approximation algorithm for
  split vertex deletion}, Proceedings of the {F}orty-{S}eventh {I}nternational
  {C}olloquium on {A}utomata, {L}anguages and {P}rogramming, 2020.

\bibitem{marx2010chordal}
D{\'a}niel Marx, \emph{Chordal deletion is fixed-parameter tractable},
  Algorithmica \textbf{57} (2010), no.~4, 747--768.

\bibitem{Oum05}
Sang{-}il Oum, \emph{Rank-width and vertex-minors}, J. Comb. Theory, Ser. {B}
  \textbf{95} (2005), no.~1, 79--100.

\bibitem{reed2004finding}
Bruce Reed, Kaleigh Smith, and Adrian Vetta, \emph{Finding odd cycle
  transversals}, Operations Research Letters \textbf{32} (2004), no.~4,
  299--301.

\bibitem{TsukiyamaIAS77}
Shuji Tsukiyama, Mikio Ide, Hiromu Ariyoshi, and Isao Shirakawa, \emph{A new
  algorithm for generating all the maximal independent sets}, {SIAM} J. Comput.
  \textbf{6} (1977), no.~3, 505--517.

\bibitem{UU2009}
Ryuhei Uehara and Yushi Uno, \emph{Laminar structure of {P}tolemaic graphs with
  applications}, Discrete Appl. Math. \textbf{157} (2009), no.~7, 1533--1543.
  \MR{2510233}

\bibitem{Yan81}
Mihalis Yannakakis, \emph{Computing the minimum fill-in is np-complete}, SIAM
  Journal on Algebraic and Discrete Methods \textbf{2} (1981).

\end{thebibliography}
\end{document}